\documentclass[journal,onecolumn]{IEEEtran}

\ifCLASSINFOpdf
\else
\fi

\usepackage{amsmath,amsfonts,amssymb}
\usepackage{graphicx}
\usepackage{amsthm}
\usepackage{multirow}
\usepackage[table]{xcolor}
\usepackage{cite}

\newtheorem{theorem}{Theorem}[section]

\newtheorem{construction}[theorem]{Construction}
\newtheorem{remark}[theorem]{Remark}
\newtheorem{definition}[theorem]{Definition}

\begin{document}
\title{Array BP-XOR Codes for Hierarchically Distributed Matrix Multiplication}

\author{Suayb~S.~Arslan,~\IEEEmembership{Member,~IEEE} % <-this % stops a space
\thanks{S. S. Arslan is with the Department
of Computer Engineering, MEF University, Maslak,
Istanbul, Turkey, e-mail: arslans@mef.edu.tr.}% <-this % stops a space
% <-this % stops a space
\thanks{The contents of the paper are partially presented in two IEEE International Symposium on Information Theory (ISIT) conferences which were held in 2018 and 2019 in Colorado, USA and Paris, France, respectively.}
\thanks{Copyright (c) 2021 IEEE. Personal use of this material is permitted.  However, permission to use this material for any other purposes must be obtained from the IEEE by sending a request to pubs-permissions@ieee.org.}}

\markboth{Accepted to IEEE Transactions on Information Theory, Peer-Reviewed and Unedited Version, Nov.~2021}%
{Arslan \MakeLowercase{\textit{et al.}}: Array BP-XOR Codes for Hierarchically Distributed Matrix Multiplication}



\maketitle

\begin{abstract}
A novel fault-tolerant  computation technique based on array Belief Propagation (BP)-decodable XOR (BP-XOR) codes is proposed for distributed matrix-matrix multiplication. The proposed scheme is shown to be configurable and suited for modern hierarchical compute architectures such as Graphical Processing Units (GPUs) equipped with multiple nodes, whereby each has many small independent processing units with increased core-to-core communications.  The proposed scheme is shown to outperform a few of the well--known earlier strategies in terms of total end-to-end execution time while in presence of slow nodes, called \textit{stragglers}. This performance advantage is due to the careful design of array codes which distributes the encoding operation over the cluster (slave) nodes at the expense of increased master-slave communication. An interesting trade-off between end-to-end latency and total communication cost is precisely described. In addition, to be able to address an identified problem of scaling stragglers, an asymptotic version of array BP-XOR codes based on projection geometry is proposed at the expense of some computation overhead. A thorough latency analysis is conducted for all schemes to demonstrate that the proposed scheme achieves order-optimal computation in both the sublinear as well as the linear regimes in the size of the computed product from an end-to-end delay perspective.
\end{abstract}


\IEEEpeerreviewmaketitle

\section{Introduction}

\IEEEPARstart{T}{oday} data science led to enormous computation and storage requirements that have transcended all expectations. The imminent consequence of this has been the distribution of data and associated computation work over large clusters of commodity processing and storage devices that are typically less capable than enterprise systems. Moreover, when these devices work together in large groups, their mean-time-to-failure can drop from a few years to a few days making their availability significantly less than the execution time of many current high-performance computing applications \cite{schroeder2007,Schroeder2010,yuan2012}. On the other hand, due to fail-stop or slow compute nodes, coined  as \textit{stragglers} in literature, the main objective of distributed computing may be compromised i.e., total end-to-end latency is severely impacted by the slowest workers in the cluster which renders the any-topology distributed computation ineffective. In a master-slave configuration, for instance, the serial portion of the overall computation may lead to a bottleneck at the master node if heavy and dependent calculations are performed. In \cite{yadwadkar2016}, it is shown that stragglers may run eight times slower than the average worker performance using
Amazon EC2 instances. Recently, fault--tolerant and efficient parallel matrix computations have gained momentum due to their immediate application to various machine learning and inference algorithms \cite{wu2011}. However, the majority of these studies primarily deals with  a faulty device problem which may lead to erroneous computation or else inevitable round-off errors accounted for by algorithmic approaches \cite{chen2008}. 
I wish you the best of success.

The primary study that tackles the stragglers (i.e., slow--performing nodes) is based on a \textit{parameter server} concept in which machine learning algorithms are shown to scale on a distributed network \cite{Li2014} in presence of fail-stop system behaviors. Later, the alternative idea of \textit{coded computation} is proposed to provide computation redundancy in robust system design against the stragglers. More specifically, in order to economically use the compute infrastructure, a coded computation framework based on a family of Maximum Distance Separable (MDS) codes is proposed in \cite{Lee2016}. We realize that coded computation can be seen as the reincarnation of algorithm-based fault tolerance for various computation tasks such as matrix--matrix multiplication \cite{Huang1984}. Later that study, the idea is exhaustively exercised in various computing tasks including large matrix multiplications, gradient computing \cite{Tandon2017} and its more recent extensions \cite{Reisizadeh2019}, convolutions \cite{Dutta2017}, optimizations \cite{karakus2017}  and Fourier transforms \cite{Yu2017F}. Despite the proposed coding scheme considering a large-scale matrix multiplication as an example, we can extend it to other types of computation that have matrix multiplication (dot products) at the core such as linear signal transformations and multi-class classifications. Particularly, the training phase of distributed machine learning algorithms are most affected by the presence of stragglers and coded computation has been investigated to solve the main trade-off between overall communication cost and service latency \cite{songze2017}. To ensure the scalability of distributed computing, the overall job execution time must be minimized.  However, most of the previous literature have focused on the  parallel task time in which the encode/decode time is excluded from the overall end-to-end latency performance. On the other hand, the main performance objective of this study would be the minimum end-to-end latency rather than measuring the pure parallel task time by allowing distributed encoding and low-complexity decoding. Rateless codes are investigated in the context of coded computing \cite{Luby2002} which naturally enables near--linear time complexity for the decoding.  However, this study is based on matrix-vector multiplication and uses Robust Soliton distribution where the maximum degree per node does not have a bound making the distributed encoding infeasible due to unacceptable communication load. Note that a distinctive feature of distributed computing is to be able to distribute encoding over the cluster nodes which is largely ignored by the past literature. In \cite{wang2018}, sparse codes are proposed which possess  similar  features so that encoding could be distributed. However, the authors mostly focused on recovery threshold, computation complexity and decoding time rather than mathematically analyzing the encoding workload. Distribution of encoding (in information-theoretic sense) is considered very recently \cite{maddah2020} where authors introduce errors in the encoding process and investigate the fundamental limits for accurate decoding and set necessary conditions. In such a setting an adversary is assumed whereas, in our scenario, no such adversary exists. Moreover, straggler mitigation  was not taken into account. 
In this paper, we propose array codes to ensure low-complexity encoding/decoding and efficient distribution of parallel as well as encoding workload over the cluster nodes. We further propose asymptotic extensions to allow for increased failures in scaling clusters.

\subsection{Linear Array Codes}

Array codes are linear codes defined for two-dimensional data structures organized typically in a matrix format including both data and parities as columns. These codes are quite attractive candidates for burst error recovery in communication \cite{dholakia2004} and distributed storage systems \cite{Farrell} and provide data reliability with optimal time/space consumption thanks to block Maximum Distance Separability (MDS) property in the code construction process. Moreover,  a great deal of work has been done and many improvements have been proposed for these codes over the years \cite{Blaum} to secure simpler math, low-complexity computations and the block MDS property all at the same time.

Typically, any linear code can be represented using a bipartite graph either using the parity check matrix or the generator matrix of the code \cite{shulin}. Using the generator matrix representation, the corresponding bipartite graph has two types of nodes: Nodes that are used to decode (check or coded nodes) and nodes that are decoded (information nodes). Nodes in bipartite graph representation are connected with edges to represent node adjacency. The neighbors of node $j$ (neighbor set), denoted by $\mathcal{N}_j$, is the set of all nodes connected to node $j$. The cardinality of the neighbor set is called the \emph{degree} of node $j$. The Belief Propagation (BP) algorithm a.k.a. message passing algorithm is a low complexity iterative decoding process (updating nodes and edges) to reconstruct data from unerased coded nodes using the sparse bipartite graph representation of the code. We begin by setting all the contents of information nodes to NULL. Then, we look for a degree-one coded node and copy the contents to its neighbor information nodes by replacing NULL. Next, we update all the coded nodes that are connected to this neighbor and eliminate the edges that established neighborhood relationships. This completes the first step, and in the next iteration, we continue applying the same methodology until there remains no information node with NULL content. If the algorithm stops prematurely during iteration, we claim a decoding failure, otherwise, we report a decoding success. Array codes have recently been studied under BP decoding \cite{Wang} and useful upper bounds are derived in \cite{Paterson2013} that theoretically establishes the relationship between the block length (and hence the rate of the code), decodability, and sparsity of the generator matrix i.e., the encoding/decoding complexity of the code.

\subsection{Motivation and Contributions}

First of all, for a given file block size, we primarily  demonstrate in this study that by relaxing the block MDS constraint on the code construction process, the previously found bounds on the code block length \cite{Paterson2013} can be relaxed while ensuring successful decoding of the file block through low complexity BP algorithm. In other words, previous works on array MDS BP-XOR codes, due to their construction, do not allow the number of failures to grow with the block length while maintaining a code rate smaller than one. Such an observation shall yield more flexible and powerful code constructions for distributed computing. For instance, a carefully chosen fixed code rate would allow dealing with linearly scaling stragglers in a large network of computing devices \cite{Schroeder2010}. Furthermore, we introduce asymptotically MDS array codes as an alternative and shall consider a discrete geometry construction based on Mojette Transform that is recently studied within the context of low density parity check codes and is shown to reduce the node repair complexity  \cite{ArslanMojLDPC}. By providing and establishing an appropriate set of code parameters, we explicitly construct codes that fulfill the desired theoretical requirements.

On the other hand, most of the existing works  have focused on small to moderate size matrix multiplication operation and thereby improving the worker node task runtime while the encoding/decoding times at the master node are assumed to be negligibly small. However, although the master node workload may be acceptable for small-scale (few tens of nodes) networks and moderately sized matrices, it shall be extremely prohibitive for large-scale (over thousands of nodes) compute tasks. Hence, the total execution time is considered to be the real optimization criterion whereby the encoding and decoding processes need to be low complexity, parallelizable and distributable. While achieving  better total execution time, the system should not lose the \textit{recovery threshold} performance due to stragglers \cite{Yu2017}. In our study, we address this issue by distributing the encoding operation over the compute cluster and allowing BP to resolve the actual matrix multiplication operation at the master. Secondly, modern compute nodes are equipped with multiple typically equal-quality cores such as Central Processing Unit (CPU) instances possessing over hundreds of physical cores or Graphical Processing Unit (GPU) instances with thousands of CUDA cores. By considering each core as a standalone network processor, and the distinctive cost nature of communications between these cores and with other network nodes, it is intuitive that coding can be exploited for optimal utilization of the underlying infrastructure.

Overall, one of the main differentiations of this study is that it focuses on end-to-end user delay rather than the pure parallel task time and clearly demonstrates advantages/disadvantages of the proposed coded computation with scaling clusters both for sublinear as well as the linear regime in the size of the product at hand. We recognize the fixed number of straggler proofing using the original array BP-XOR code constructions and further proposed to use asymptotical versions to address scaling number of stragglers. In addition, we assumed a hierarchically clustered compute architecture with varying degrees of parallelism which aligns with the realistic infrastructures built today with Google and Amazon virtualized environments. Furthermore, the proposed coding scheme allows distributing encoding workload over the cluster nodes at the expense of more bandwidth utilization opening up new options in the computation-bandwidth trade-off space. For the sake of generality, we consider generic matrix sizes instead of square matrices.  

The rest of the paper is organized as follows. In Section \ref{SectionCC}, system model is introduced.  In Section \ref{AMDSBPXOR}, array BP-XOR codes are primarily introduced for distributed computation, and then, its asymptotical version is formally defined. A discrete geometry construction is proposed to enhance the achievable code rate of the classic case. Section \ref{LatencyPerf} analyzes the end-to-end average latency performance as well as communication costs of the proposed class of codes in a distributed coded computation setting and Section \ref{numerical} provides numerical results to support our arguments. Finally,  Section \ref{conc11} concludes the paper. Some of the proofs are intentionally omitted for the smooth flow of the paper and later provided in appendices.

\section{Coded Computation System Model} 
\label{SectionCC}

Let us consider the multiplication of two large matrices i.e., $\mathbf{A}^\intercal \mathbf{B}$  where  $\mathbf{A} \in \mathbb{F}^{s \times k}$ and $ \mathbf{B} \in \mathbb{F}^{s \times b}$ and $\mathbb{F}$ denotes any algebraic field. For a generic matrix $\mathbf{X}$, we use $\mathbf{x}_{i,:}$ to denote the $i$-th row and $\mathbf{x}_{:,j}$ the $j$-th column for the rest of the paper. Thus, computing $\mathbf{A}^\intercal \mathbf{B}$ amounts to $kb$ dot products (each containing $s$ multiplications and $s-1$ additions). In our system, the \textit{Generator} unit (a.k.a. the master node) either communicates individual rows/columns of $\mathbf{A}$ and $\mathbf{B}$ or their sums to $n$ compute nodes and corresponding dot products will be performed by $n \geq k$ nodes each equipped with $M \geq b$ processors in which only $b$ processors are assigned to task execution by the local scheduler. Note that if $M < b$, each processor will have to execute more than one dot product. In that case, we group (partition) rows of $\mathbf{A}$ and columns of $\mathbf{B}$ in such a way to equally distribute computation on the available processors in the network. The following remark presents a particular partitioning strategy that might be needed to achieve the constraint $b<nM$.

\begin{remark}
For any given integer $m$ satisfying the relation $M>m>0$, we divide the columns of $\mathbf{A}$ in to $m$ equal size ($s \times k/m$) matrices and the columns of $\mathbf{B}$ in to $m$ equal size ($s \times b/m$) matrices. In that case, $m^2$ processors multiply two matrices of sizes $k/m \times s$ and $s \times b/m$. Note that we do not guarantee $m^2 < nM$ i.e., total number of submatrix computations being less than total number of processors without an additional constraint. Suppose for instance $m|b$ and $b|k$. With that constraint, we can multiply fractions of matrices in each node whereby each processor unit performs small matrix multiplications of sizes $mk/b \times s$ and $s \times m$ where $s$ is typically assumed to be large. In that case, we have $b$ matrix multiplications of sizes $mk/b \times b/m$ implying $b \leq nm < nM$.
\end{remark}

Since core-to-RAM communication is orders
of magnitude faster than node-to-node communication, we consider an appropriate multi-processor setting that takes into account this observation. The summary of the coded master-slave compute cluster architecture is illustrated in Fig. \ref{fig:cc1}. Note that a similar multi-processor setting is considered in \cite{LeeRam2017} where non-linear local functions are computed with a fixed and equal number of local cores per node.  

\begin{figure}[t!]
\centering
  \includegraphics[width=\linewidth]{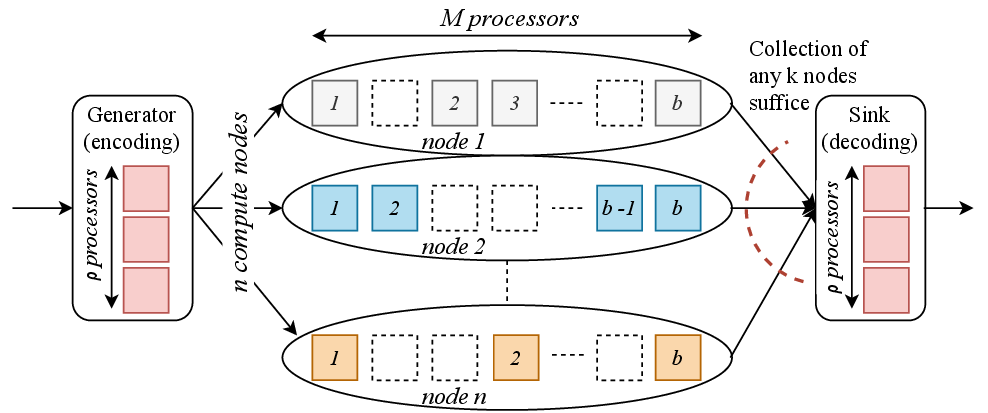}
  \caption{A clustered compute architecture is considered in this study. Master's processors are more capable and perform operations $c$ times faster than the node processors. Each node is equipped with $M$ processors/cores and only a subset of size $b$ processors/cores are typically assigned to the computation. The same architecture can be used to perform coded computation in case asymptotic extensions of MDS Array BP-XOR codes are used where $M=\max\{b_i\}$ may be chosen with $M \geq b$.}
  \label{fig:cc1}
  \vspace{-5mm}
\end{figure}

In this setting, the matrices being operated on are provided as inputs to the Generator unit  where the encoding operation typically takes place on $\rho$ processors. However, the encoding operation can alternatively be distributed over the cluster nodes if the associated coding structure allows so\footnote{For instance as it will be demonstrated with this study, array codes based on pure XOR logic that can be decoded using belief propagation can provide such flexibility at the expense of increased communication cost.} as shall be shown with the help of an example in the next section. Encoded rows/columns of the matrices are communicated with the nodes of the cluster in which a total of $bn$ processors compute the matrix--matrix multiplication together. We assume the cluster processors are slower than that of the master node by a factor of $c \geq 1$. Finally, the \textit{Sink} unit collects a subset of  processor outputs to initiate decoding i.e., putting together the final product in place.  Generator and Sink are not necessarily two physical nodes as drawn in this figure. Rather, they are abstracted units that may reside in the same physical (master) node. By construction, having all processor outputs of any $k$ out of $n$ nodes will suffice to reconstruct $\mathbf{A}^\intercal \mathbf{B}$. 

As shown in Fig. 1, the total execution time in our clustered distributed setting is given by the sum of encoding time $T_{e}$, master-slave transmission time $T_{ms}$, the overall slave task time where each slave completes its execution with time $T_i$, slave-master transmission time $T_{sm}$ and finally the decoding time at the master denoted by $T_{d}$. Hence, we can express the overall latency similar to \cite{baharav2018} and \cite{Suayb2019} as
\begin{eqnarray}
T_{total} = T_e + T_{ms} + \min_{S \in \mathcal{S}} \max_{i \in S} T_i + T_{sm} + T_{d} \label{eqnTtotal}
\end{eqnarray}
where $\mathcal{S}$ is the  set of all minimal decodable subsets of $\{1,2,\dots,nb\}$ processor outputs. Most of the past work focuses on minimizing the parallel task time while paying little if no attention to encode/decode times. Plus, some of the previous works favor parallel task execution times at the expense of consuming more bandwidth and master system CPU resources \cite{Fahim}. However, as the system size (as well as the matrix sizes) scales, the encode/decode times of the master node will become the main bottleneck of the overall system performance. In our work, $T$ will represent the sum of computation times, i.e., $T = T_e + \min_{S \in \mathcal{S}} \max_{i \in S} T_i  + T_{d}$ which will constitute the main focus of this paper.

\section{Asymptotically MDS Array BP-XOR Codes for Matrix Multiplication}
\label{AMDSBPXOR}

Before defining the class of asymptotically MDS array BP-XOR codes, let us provide the conventional definition of MDS BP-XOR codes using the notation of  \cite{Paterson2013}. Accordingly, let $\mathrm{l}$ be the symbol size in bits and $\mathcal{M} = \{0,1\}^\mathrm{l}$ be the symbol set from which we select our information as well as coded symbols. The fundamental operation we use is the Exclusive OR (XOR). In our study, nodes represent blocks of data that contain one or more symbols in it. Symbols are the smallest data unit over which XOR operations are defined. 

\subsection{Array BP-XOR codes for Coded Computation}

An $[n,k,t,b]$ array BP-XOR code is a $b \times n$ two dimensional rate $r=k/n$ binary linear code $\mathcal{C} = [a_{i,j}]_{1 \leq i \leq b, 1 \leq j \leq n}$ in which the coding symbol $a_{i,j} \in \mathcal{M}$ is the XOR of a subset of source symbols $I = \{v_1,\dots,v_{bk}\}$, typically structured as a $b \times k$ data matrix, and $I$ can be reconstructed from any $n-t$ columns of the linear code $\mathcal{C}$ using BP algorithm for an appropriate integer $t \leq n-k$.  The degree of a coded symbol $a_{i,j}$, denoted as $\sigma_{i,j}$ ( $\leq \sigma$, the maximum node degree number of the code), is the number of information symbols that participate in logical XOR operation i.e., $a_{i,j} = v_{z_1} \oplus \dots \oplus v_{z_{\sigma_{i,j}}}$ such that $v_{z_s} \in I$ for all $s \in \{1,\dots,\sigma_{i,j}\}$. A $t$-erasure correcting array BP-XOR code is block MDS if the source symbols can be reconstructed from $k = n-t$ columns of $\mathcal{C}$.

Just like product codes \cite{baharav2018}, \cite{Park2019}, we encode computation in two dimensions however with the exception that the encoding is not only in vertical and horizontal directions, both could be in any carefully chosen direction, which shall provide more flexibility between the distribution of encoding and the use of bandwidth. In addition, every computation task of the product-coded scheme involves only a single dot product whereas our scheme performs a maximum of $\sigma \geq 1$ dot products per processor yielding $b\sigma$ dot products per node in the worst case. However, in the product-coded scheme, the master node performs heavy and unbalanced encoding operations which will be shown to be a bottleneck from an overall latency point of view. With the array BP-XOR codes, encoding operation may be distributed over the clustered computation network at the expense of increased bandwidth consumption between the master node and the processors of the cluster nodes. On the other hand, distributing the decoding task among the network nodes is still an open research topic.

\textbf{Example 1:} To illustrate specifically the assigned jobs for each node, let us suppose we would like to compute (with $b = m = 2$) the following simple matrix multiplication,
\[
\mathbf{A}^\intercal \mathbf{B} =
\begin{bmatrix}
    \mathbf{a}_{1,:}   \\ \mathbf{a}_{2,:} 
\end{bmatrix}
\begin{bmatrix}
    \mathbf{b}_{:,1}  &
    \mathbf{b}_{:,2} 
\end{bmatrix} 
= 
\begin{bmatrix}
    v_1  & v_2 \\
    v_3 & v_4 
\end{bmatrix} 
\]
by using $[5,2,2,2]$ MDS array  BP-XOR code with $\sigma=2$, given in Table \ref{table1} with the designated computation distribution among $n=5$ nodes each equipped with two processors. For instance, node 1 receives the entire matrices $\mathbf{A}$ ($\mathbf{a}_{1,:}$, $\mathbf{a}_{2,:}$) and $\mathbf{B}$ ($\mathbf{b}_{:,1}, \mathbf{b}_{:,2}$) to be able to compute $v_1$ on its processor 1, $v_2$ and $v_3$ and hance $v_2+v_3$ on its processor 2. On the other hand, node 5 gets $\mathbf{A}$ ($\mathbf{a}_{1,:}$, $\mathbf{a}_{2,:}$) and $ \mathbf{b}_{:,1} +  \mathbf{b}_{:,2} $ to compute $v_1+v_2$ on its processor 1 and $v_3+v_4$ on its processor 2 which reduces the bandwidth compared to what node 1 gets. 

\begin{table}[t!]
\centering
\begin{tabular}{|c|c|c|c|c|c|}
\hline
& node 1 & node 2 & node 3 & node 4 & node 5 \\ 
 \hline\hline
processor 1 & $v_1$ & $v_2$ & $v_3$ & $v_4$ & $v_1 + v_2$ \\ \hline
processor 2 & \cellcolor{gray!25} $v_2 + v_3$ & \cellcolor{gray!25} $v_1 + v_4$ & $v_2 + v_4$ & $v_1 + v_3$ & $v_3 + v_4$ \\
\hline
\end{tabular}
\vspace{1mm}
\caption{[5,2,2,2] MDS array BP-XOR Code for $\sigma=2$} \label{table1}
\vspace{-2mm}
\end{table}

From these examples, we observe based on the code construction that encoding can optionally be carried out in the master or the cluster nodes. However, in some of the subcomputations such as $v_2+v_3$ or $v_1+v_4$ (colored gray in Table \ref{table1}), the encoding has to be distributed. On the other hand, for the rest of the subcomputations, the encoding can be distributed among other compute nodes at the expense of more communication. However, if we choose to minimize the bandwidth, the master could help with the encoding by executing row or column additions. In Section IV.B, to be able to simplify the system description we assume the worst case scenario in terms of communication cost i.e., $\sigma$ times larger than that of competent schemes, and let all encoding processes take place on cluster nodes. 

Note that with this setting as soon as any two out of five nodes complete their processing, the master node will initiate a belief propagation decoding to put together $\mathbf{A}^\intercal \mathbf{B}$. Although in this case four-processor outputs of any selection of two nodes are sufficient to reconstruct the result, we have to wait for all ingredient processors (due mainly to block MDS property) to finish their execution. On the other hand, the recovery threshold of this code can be shown to be 7, i.e., any 7 processor executions will be sufficient to reconstruct the result in the worst case. 

For comparison purposes, let us also give an example for a polynomial code \cite{Yu2017} with similar parameters i.e., $m=2$ and $k=4$. Let us define $a_j := \mathbf{a}_{1,:} + j\mathbf{a}_{2,:}$ and $b_j := \mathbf{b}_{:,1} + j^2\mathbf{b}_{:,2}$.  With this definition, the node 1 shall receive $a_1$, $a_6$, $b_1$ and $b_6$ and compute the dot product $a_1b_1$ on processor 1 and $a_6b_6$ on processor 2 as shown in Table \ref{table2}. The rest of the nodes receive the same amount of data and executes exactly two products. Note that the distribution of the tasks can be done in any order as the computation of each task has the same complexity. Unlike MDS array BP-XOR codes, node processors in this case computes a single dot product provided that the master first encodes and generates $a_j$s and $b_j$s. On the other hand, due to the sparse nature of  MDS array BP-XOR codes, some nodes (such as node 5) consumes less bandwidth at the expense of more dot product computations. Also the encoding operation at the master is simpler compared to polynomial codes. The recovery threshold of polynomial codes is given by $m^2 = 4$ achieving the minimum possible.

\textbf{Example 2:} Let us suppose we would like to compute (with $b = m = 3$) the following matrix multiplication,
\[
\mathbf{A}^\intercal \mathbf{B} =
\begin{bmatrix}
    \mathbf{a}_{1,:}   \\ \mathbf{a}_{2,:} \\ 
    \mathbf{a}_{3,:} 
\end{bmatrix}
\begin{bmatrix}
    \mathbf{b}_{:,1}  &
    \mathbf{b}_{:,2} &
    \mathbf{b}_{:,3} 
\end{bmatrix} 
= 
\begin{bmatrix}
    v_1  & v_2 & v_3 \\
    v_4 & v_5 & v_6 \\
    v_7 & v_8 & v_9
\end{bmatrix} 
\]
by using $[5,3,2,3]$ MDS array  BP-XOR code with $\sigma=2$, as given in Table \ref{table33} with the designated computation distribution among $n=5$ nodes each equipped with three processors/cores. As can be seen, computations in gray cells do not share any common terms from matrices $\mathbf{A}$ and $\mathbf{B}$. Thus, encoding for these computations has to be distributed. Note that the ratio of required distributed encoding is 6/15 = 0.4 for this code whereas it was 2/10 = 0.2 for the code given in Table \ref{table1}, showing the natural dependency on the code construction.

\begin{table}[t!]
\centering
\begin{tabular}{|c|c|c|c|c|c|}
\hline
& node 1 & node 2 & node 3 & node 4 & node 5 \\ 
 \hline\hline
processor 1 & $a_1b_1$ & $a_2b_2$ & $a_3b_3$ & $a_4b_4$ & $a_5b_5$ \\ \hline
processor 2 & $a_6b_6$ & $a_{7}b_{7}$ & $a_8b_8$& $a_9b_9$ & $a_{10}b_{10}$ \\
\hline
\end{tabular}
\vspace{1mm}
\caption{A (10,4) Polynomial MDS Code. %Any order other than the one shown is valid.
} \label{table2}
\vspace{-2mm}
\end{table}

\begin{table}[t!]
\centering
\begin{tabular}{|c|c|c|c|c|c|}
\hline
& node 1 & node 2 & node 3 & node 4 & node 5 \\ 
 \hline\hline
processor 1 & $v_1$ & $v_1 + v_2$ & $v_2 + v_3$ & $v_7$ & $v_3$ \\ \hline
processor 2 & \cellcolor{gray!25}  $v_3 + v_4$ & $v_4 + v_5$ & $v_5 + v_6$ & \cellcolor{gray!25} $v_9 + v_1$ & \cellcolor{gray!25}  $v_2+v_6$ \\
\hline
processor 3 & \cellcolor{gray!25}  $v_6 + v_7$ & $v_7 + v_8$ & $v_8 + v_9$ & \cellcolor{gray!25}  $v_4 + v_8$ & \cellcolor{gray!25}  $v_9 + v_5$ \\
\hline
\end{tabular}
\vspace{1mm}
\caption{[5,3,2,3] MDS array BP-XOR Code for $\sigma=2$} \label{table33}
\vspace{-5mm}
\end{table}

\subsection{Asymptotical Extensions}

Let $\sigma$ to be the maximum node degree of a given array BP-XOR code, we note from \cite{Paterson2013} that if $k = \sigma$ it is not hard to show that 
\begin{eqnarray}
  n &\leq& kb + 1 + \max\{k-3,0\}
\end{eqnarray}
from which we easily deduce that the upper bound on $n$ can be arbitrarily large (i.e., for $b \gg 1$) and allow any arbitrarily small code rate $r$ to be possible. However, for $k > \sigma$ it is observed that the array code blocklength $n$ is upper bounded based on a specific choice of $k$ \cite{Paterson2013}. In addition, we observe from the same study that for $b \gg 1$ and large enough $k$ i.e., $k > \sigma^2$ we have $n \leq k + \sigma -1$. This also implies that for a large enough information block length $k$, the achievable rate will be close to 1, putting a constraint on the code design rate. This ultimately means that the portion of straggler tolerance would not be scaling well with the size of the cluster. By fixing $\sigma$, we shall control the complexity of encoding/decoding processes and as we shall see in later sections the overall end-to-end computation latency. We considered an asymptotic extension  of such codes next to allow better flexibility in terms of choosing the right code rate for the given coded computation system (i.e., scaling number of stragglers) at the expense of using slightly more processors work per node (Remember that we assume to have $M \geq b$ cores per node to respond to such requirement if need be). 

For a given positive integer $b^\prime$ satisfying $b^\prime > b$, a $[n,k,t,b, b^\prime]$ asymptotically MDS array BP-XOR code $\mathcal{C}^a$ is a linear code with $i$-th column $(y_{i,1}, \dots, y_{i,b_i}) = (x_{1}, \dots, x_{bk})G_i$ for a $bk \times b_i$ generator matrix $G_i, i \in \{1,\dots,n\}$ such that $b^\prime = (1/n)\sum_ib_i$. Thus, the generator matrix for $\mathcal{C}^a$ is given by the $bk \times \sum_i b_i$ matrix,
\begin{eqnarray}
G_{\mathcal{C}^a} = [G_1 | G_2 | \dots | G_n].
\end{eqnarray}

A $t$-erasure correcting asymptotically MDS array BP-XOR can perfectly reconstruct the data matrix $I$ from any $n-t=k$ column combinations of $\mathcal{C}^a$ using BP decoding and as $b \rightarrow \infty$ we have $b^\prime \rightarrow b$. Note that the raw source data need not be in standard $b \times k$ form. For any positive integer $g$ satisfying $b|g$ and $k|g$, the generator matrix $G_{\mathcal{C}^a}$ should work fine for different arrangements of the data block matrix such as $b/g \times kg$. We finally note that the code $\mathcal{C}^a$ is not in two-dimensional standard rectangle form as in $\mathcal{C}$. However, we introduced another parameter $b^\prime$ to be able to make asymptotically MDS array BP-XOR codes analogous to standard MDS array codes defined over rectangle shape binary matrices. 

For a given fixed code rate $r$ and $n$, let us define $\epsilon(b,n)$ to be the maximum coding overhead\footnote{Since columns of $\mathcal{C}^a$ may have different sizes, the overhead depends on which $k$ columns are used for reconstruction. Thus, $\epsilon(b,n)$ is the maximum over all combinations of $k$ columns.} Eventually, the coding overhead depends on the number of columns $n$ in the code, so called array code blocklength. of $\mathcal{C}^a$ satisfying $b^\prime = (1 + \epsilon(b,n))b$. The asymptotically optimal overhead property implies that as $\epsilon(b,n) \rightarrow 0$ we have $b\rightarrow \infty$. Let us provide the following theorem that sets the necessary condition/s on the parameters for the existence of asymptotically MDS array BP-XOR codes.

\begin{theorem} \label{Thm21}
Let  $\mathcal{C}^a$ be a $[n,k,t,b, b^\prime]$ asymptotically MDS array BP-XOR code such that the maximum coded node degree $\sigma$ satisfies $2 < \sigma < (bk-1)/(b^\prime-1)$. Then, we have
\begin{eqnarray}
n &\leq& k + \sigma - 1 \dots \nonumber \\
&& \ \  + \left\lfloor \frac{b(k(\sigma^\prime - \sigma) + (\sigma-1)\sigma^\prime) - (\sigma-1)(3\sigma/2 - 1)}{b(k - \sigma^\prime) + \sigma - 1} \right\rfloor \nonumber
\end{eqnarray}
where $\sigma^\prime = \sigma(1 + \epsilon(b,n))$ and $\epsilon(b,n)$ is the maximum coding overhead.
\end{theorem}

\begin{proof}
Since the code is assumed to be block MDS, i.e., it is able to tolerate $n-k$ column erasures of $\mathcal{C}^a$,  each information symbol $v_s \in I$ must appear in at least $n-k+1$ columns, otherwise information symbols cannot be reconstructed. Since there are $kb$ symbols, we shall have
\begin{eqnarray}
kb(n-k+1) \label{eqn1}
\end{eqnarray}
minimum symbol appearances in $\mathcal{C}^a$. On the other hand, we observe that belief propagation decoding needs to have degree-one encoding symbols to start decoding. So we need at least $n-k+1$ degree-one symbols in distinct columns of  $\mathcal{C}^a$ (in the worst case of $n-k$ column erasures when each may comprise one degree-one symbol). Similarly, we need at least one degree-two, one degree-three, $\dots$, one degree-$(\sigma-1)$ coding symbols to make sure that BP decoding continues. Although it is possible to have multiple degree-two symbols and continue BP decoding, by making this choice we attempt to maximize the appearance of information symbols in $\mathcal{C}^a$. Note that if these symbols happen to be in distinct unerased columns, the bound could be tightened, otherwise the bound might still be loose for say if $ \sigma > k + 1$ which is not usually typical. Therefore, in such a formulation a total of $(n-k+1)+\sigma-2 = n-k+\sigma-1$ symbols are assigned degrees. The rest of the $b^\prime n - (n-k+\sigma-1)$ symbols can have at most $\sigma$ degree. Thus, $C^a$ can have at most
\begin{eqnarray}
\sigma(b^\prime n- (n-k+\sigma-1)) + n - k + 1 + \sum_{i=2}^{\sigma-1} i 
\end{eqnarray}
or equivalently,
\begin{eqnarray}
\sigma(b^\prime n- (n-k+\sigma-1)) + n - k  +  \frac{\sigma(\sigma-1)}{2} \label{eqn2}
\end{eqnarray}
appearances of $kb$ information symbols. We can rewrite (\ref{eqn2}) in a more compact form as
\begin{eqnarray}
\sigma b^\prime n - (\sigma-1)(n-k+\sigma/2)
\end{eqnarray}

Using the natural relation $(\ref{eqn1}) \leq (\ref{eqn2})$, and assuming we have $b(k - \sigma^\prime) + \sigma - 1 > 0$, we can collect all terms that includes $n$ on the left hand size and find an upper bound on $n$ as follows,
\begin{eqnarray}
kb(n-k+1) &\leq& \sigma b^\prime n - (\sigma-1)(n-k+\sigma/2) 
\end{eqnarray}
which leads to 
\begin{eqnarray}
n &\leq& \left\lfloor \frac{(kb+\sigma - 1)(k-1) - (\sigma - 1)(\sigma/2-1)}{b(k - \sigma^\prime) + \sigma - 1 } \right\rfloor \label{eqn3} \\
&=& k + \sigma - 1 \dots \nonumber \\ 
&& \ \ \ \ + \left\lfloor \frac{b(k(\sigma^\prime - \sigma)  (\sigma-1)\sigma^\prime) - (\sigma-1)(3\sigma/2 - 1)}{b(k - \sigma^\prime) + \sigma - 1} \right\rfloor \nonumber
\end{eqnarray}
where $\sigma^\prime = \sigma(1 + \epsilon(b,n))$ which completes the proof. Note that if $b \rightarrow \infty$ we will have $\sigma^\prime \rightarrow \sigma$ and hence equation (\ref{eqn3}) becomes identical to equation (1) of \cite{Paterson2013} i.e.,
\begin{eqnarray}
n \leq k + \sigma + 1 + \left\lfloor \frac{\sigma(\sigma-1)(b-1)}{(k-\sigma)b+\sigma-1} \right\rfloor
\end{eqnarray}
except the term $(\sigma-1)(\sigma/2-1)$. This term is essentially what makes the upper bound  improved (tighter).
\end{proof}

There are two cases that are interesting to consider for understanding the asymptotical performance. First, if $b$ tends large we will have $\sigma^\prime \rightarrow \sigma$. Hence,
\begin{eqnarray}
n &\leq& k + \sigma - 1 + \left\lfloor \frac{(\sigma-1)\sigma}{k - \sigma} \right\rfloor - \textbf{1}_{(k-\sigma)|(\sigma-1)\sigma} \label{inequality}
\end{eqnarray}
where $\textbf{1}_{A}$ is logical one if $A$ is true, otherwise it is zero. This indicator function is used due to the flooring operation and $\sigma$ only equals to $\sigma^\prime$ in the limit. Thus, if the code becomes array MDS in the limit, there remains no dependence of $n$ on $b$. On the other hand, if we let large but fixed $b \leq k$, and if $k$ gets large, we shall have
\begin{eqnarray}
n \leq k + \sigma^\prime - 1 
= k + \sigma(1 + \epsilon(b,n)) - 1 \label{eqn9}
\end{eqnarray}
which can be made arbitrarily large if we choose an appropriate $\epsilon(b, n) > 0$ for a fixed $b$ and large $n$. This essentially demonstrates that as the array BP-XOR code becomes near-optimal in terms of recovery performance, the upper bound on the number of code columns $n$ can dramatically be improved. 

Although the desirable properties of the coding overhead are found, we still need specific constructions to quantify or bound the coding overhead and hence present tighter bounds on $n$ (and $r$) for a specific construction. Based on this observation, we shall present a code construction method that uses the result of Theorem \ref{Thm21} and has an appropriate $\epsilon(b,n)$ with the properties as summarized with the following remark.

\begin{definition}
For a given fixed code rate $r=k/n$ and $n$, let $\epsilon(b,n)$ to be the maximum coding overhead of $\mathcal{C}^a$ satisfying $b^\prime = (1 + \epsilon(b,n))b$. The asymptotically optimal (MDS) overhead property implies that as $b\rightarrow \infty$ we have $\epsilon(b,n) \rightarrow 0$.
\end{definition}

\subsection{Discrete Geometry Constructions of Asymptotically-MDS array BP-XOR codes}
\label{SectionMoj}

In this section, we will introduce a particular construction of asymptotically MDS array BP-XOR codes based on discrete geometry \cite{Moj} and show that they can be regarded as a special type of the class of asymptotically MDS BP-XOR codes.

The discrete geometry construction is also known as Mojette Transform (MT) codes which are based on discrete version of Radon Transform \cite{Radon}, and can be used to generate redundancy not just for rectangle two dimensional data grid but also for any convex data grid. In  our study, we consider matrix (rectangle) data for simplicity and let encoder compute a linear set of projections at angles specified by a couple of coprime integers $(p,q)$ (with $gcd(p,q)=1$) from a $k \times b$ discrete data structure $f:(z,l) \rightarrow \mathbb{N}$. Suppose that we generate $n$ projections with parameters $\{(p_i, q_i), 0 \leq i \leq n-1\}$. The length of the projection $i$, denoted by $b_i$, is a function of angle parameters $(p_i,q_i)$ and the data grid size $k \times b$. It can be expressed in a closed form as follows \cite{Moj},
\begin{eqnarray}
b_i =  |p_i| (k-1) +  |q_i| (b-1) + 1
\end{eqnarray}

\begin{figure}[t!]
\centering
\includegraphics[angle=0, height=70mm, width=110mm]{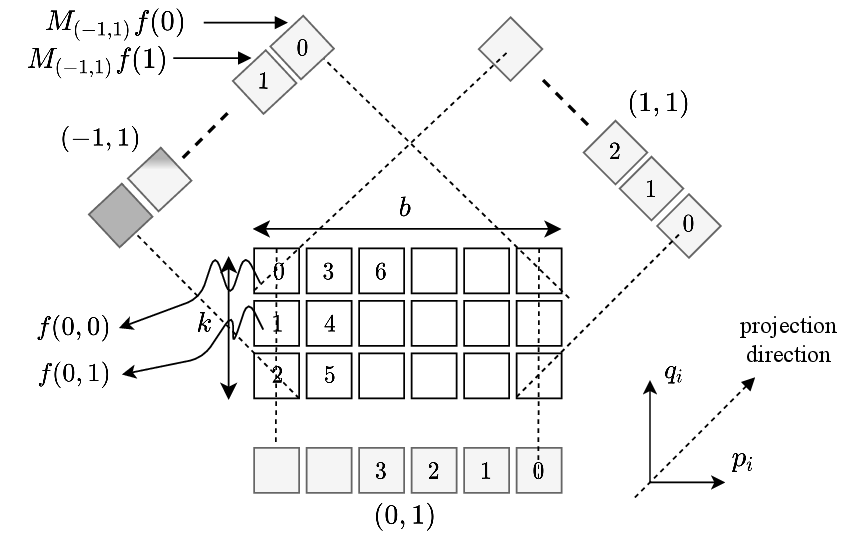}
\caption{A simple illustration of the projection concept and Mojette Transform coding with three projections with parameters $(-1,1),(0,1),(1,1)$. }\label{fig:BSW}
\end{figure}

Note that in this construction, generated projections can be treated as the columns of the asymptotically-MDS BP-XOR code.  Each bin or symbol of the $i$-th projection, based on $(p_i,q_i)$, can be computed as given by the following compact formulation
\begin{align}
& M_{(p_i,q_i)}f(m + (b-1)q_iu(q_i) +(k-1)p_iu(p_i)) \nonumber \\ & \ \ \ \ \ = \bigoplus_{z=0}^{b-1}\bigoplus_{l=0}^{k-1} f(z,l) \delta_{m  + zq_i + lp_i}\label{moj1}
\end{align}
for all $m$ satisfying the inequality,
\begin{gather*}
- (b-1)q_iu(q_i) -(k-1)p_iu(p_i) \\
 \leq m \leq \\ b_i  - (b-1)q_iu(q_i) -(k-1)p_iu(p_i) - 1
\end{gather*}
where $M_{(p_i,q_i)}$ is the transformation operator acting on $f$,  $\bigoplus$ stands for Boolean XOR operation, $u(.)$ is the discrete unit function and  $\delta_i$ is Kronecker delta function which are respectively given by
\[
u(t)=
\begin{cases}
1, & \textrm{if } t > 0 \\
0, & \textrm{Otherwise }
\end{cases}
,\mathrm{and} \ \ \
\delta_i=
\begin{cases}
0, & \textrm{if }i\not=0 \\
1, & \textrm{if }i = 0
\end{cases}
\]
An example code with parameters $k = 3$, $b = 4$ with $n = 3$ projections with parameters $(-1,1),(0,1),(1,1)$ is shown in Fig. \ref{fig:BSW}. Also shown in the same figure transformed symbols $M_{(-1,1)}f(0), M_{(-1,1)}f(1), \dots$ etc. which are the symbols of the projection with $(p_i,q_i) = (-1,1)$. MT codes can be decoded using BP algorithm and the exact reconstruction of user data matrix is possible if the projection parameters $(p_i,q_i)$ are selected judiciously according to the following Katz criterion.

\begin{theorem}[Katz Criterion\cite{katz}] \label{katz}
For a given asymptotically-MDS BP-XOR code defined by $n$ projections on a $k \times b$ data matrix where only $t \in\{1,\dots,n\}$ projections with parameters $(p_i, q_i)$ are available. Exact data reconstruction is possible using iterative BP if
\begin{align}
\sum_{i=0}^{t-1} |p_i| \geq b \textrm{ or } \sum_{i=0}^{t-1} |q_i| \geq k
\end{align}
\end{theorem}

\begin{proof}
This can be interpreted as the reconstruction of a rectangle grid using inverse Mojette transformation of projections \cite{Kingston2014}. It is not hard to see that this reconstruction technique is identical to the belief propagation algorithm for erasure recovery which was applied to discrete tomography and image reconstruction in the past \cite{gouillart2013}. The condition that ensures reconstruction is known as the Katz criterion where the full proof can be found in \cite{katz}.
\end{proof}

\begin{theorem}
If $\sigma_i, i \in \{1,2,\dots,n\} $  denotes the maximum degree of the $i$th projection defined by the parameters $(p_i, q_i)$. We have $\sigma_i = \min\{\lceil b/|p_i| \rceil, \lceil k/|q_i| \rceil \}$ and hence $\sigma = \max_i\{\sigma_i\}$.
\end{theorem}

\begin{proof}
 Considering the equation (\ref{moj1}) and the worst case scenario, we would like to find the number of $l$ and $z$ values such that $zq_i + lp_i = -m$. It is not hard to see that the maximum number of $z$ values that can satisfy this equation is given by $\lceil k/|q_i| \rceil$ due to  $0 \leq z \leq k-1$. Similarly, the maximum number of $l$ values  that can satisfy this equation is given by $\lceil b/|p_i| \rceil$ due to  $0 \leq l \leq b-1$. Since the number of possibilities for $z$ and $l$ are also constrained by the two dimensional rectangle shape, we have the maximum encoding symbol degree equal to the minimum of the two i.e., $\sigma_i = \min\{\lceil b/|p_i| \rceil, \lceil k/|q_i| \rceil \}$. Thus, the maximum degree of all the code symbols is given by the maximum degree of all the projections i.e., $\sigma = \max_i \{ \min\{\lceil b/|p_i| \rceil, \lceil k/|q_i| \rceil \} \}$. 
\end{proof}

Next, we quantify the coding overhead for MT-based asymptotically MDS BP-XOR codes by considering $k=\sigma$ and $k > \sigma$ cases separately. 

\subsubsection{ Case $k = \sigma$} First of all, note that depending on the choices of $(p_i, q_i)$, the coding overhead as well as the maximum degree of the code can change. Although there are multiple choices for  $k = \sigma$, we provide the typical choice below that also ensures good coding overhead.

\begin{construction} \label{Cons33}
Let us consider the following choice of coprime integers,
\begin{align}
q_i = 1, p_i &\in \mathfrak{T} = \left\{-\left\lfloor\frac{n-1}{2}\right\rfloor,\dots,-1,0,1,2,\dots,\left\lceil\frac{n-1}{2}\right\rceil\right\} \label{option1}
\end{align}
where $\mathfrak{T}$ is known as canonical enumeration of integers \cite{OEIS} that goes with the name \emph{A007306} and satisfies $gcd(p_i, q_i) = 1$ for $i=0,\dots,n-1$. 
\end{construction}

Note that this construction satisfies the Katz criterion simply because collecting any $k$ projections will lead us to have $\sum |q_i| = k$. If we use the coprime integers as given by the Construction \ref{Cons33}, we have $q_i$ never equal to zero and $\sigma_i = \min\{\lceil b/ \lceil (n-1)/2 \rceil,k\}$. We note that we have $\sigma = k$ for  $b \gg 1$. We next quantify the coding overhead for this particular construction and show the asymptotically optimal property.

\begin{theorem} \label{Thm37}
For the Mojette code with parameters as given in Construction \ref{Cons33}, for $b \gg 1$, we have
\begin{eqnarray}
\epsilon(b,n) \leq \frac{n(2-r)(nr-1)}{4b} + o(1)
\end{eqnarray}
where $r$ is the fixed rate of the array BP-XOR code.
\end{theorem}

\begin{proof}
See appendix \ref{AppendixA} for the proof of this theorem.
\end{proof}

\begin{figure*}
  \centering
  \includegraphics[width=0.5\columnwidth, height=6.5cm]{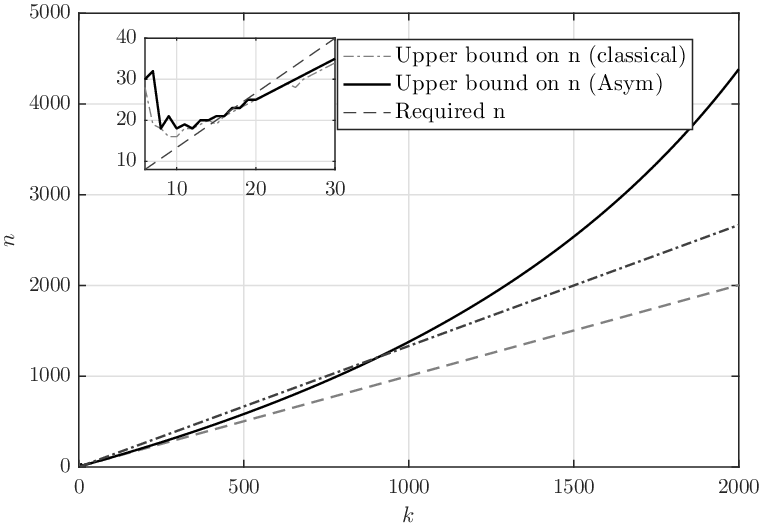}
  \includegraphics[width=0.5\columnwidth, height=6.5cm]{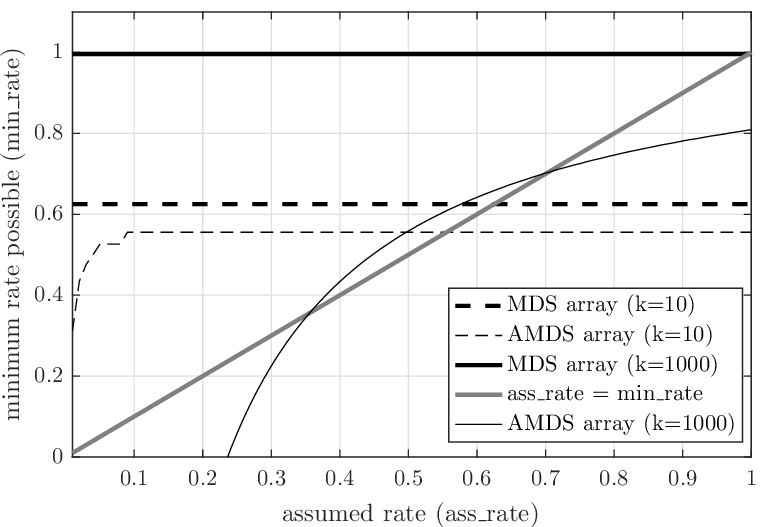}
  \caption{a) Upper bounds on the parameter $n$ are shown as a function of $k$ for rate $3/4$. b) Achievable code rates for different $k$ values are demonstrated. min\_rate: minimum rate possible, ass\_rate: assumed rate.}
  \label{fig3}
\end{figure*}

For fixed $r$ and $k$ (i.e., fixed $n$), if $b \rightarrow \infty$ then it is clear that $\epsilon(b,n) \rightarrow 0$ proving the asymptotical property. On the other hand, for fixed $r$ and $b$, if $n \rightarrow \infty$ then we have $\epsilon(b,n) \rightarrow \infty$. In fact, it is not hard to see that $\epsilon(b,n) = O(n^2)$. Therefore, due to these desirable properties of the overhead and considering the inequality (\ref{eqn9}), we can make $n$ arbitrarily large. Particularly we can find the following lower bound on $n$ for $k = rn = \sigma$ and $r < 1/2$,
\begin{eqnarray}
n \leq rn + rn \left(1 + \frac{n(2-r)(nr-1)}{4b} \right) - 1
\end{eqnarray}
which yields the inequality
\begin{eqnarray}
n-2nr \leq \frac{n^3r^2(2-r)}{4b} 
\Rightarrow n \geq \sqrt{\frac{4b(1-2r)}{r^2(2-r)}}
\end{eqnarray}

This final lower bound shows that the value for the block length $n$ can be arbitrarily large for judiciously selected large $b$. Note that the case $k = \sigma$ has the least constraint on the code block length for any MDS array BP-XOR code. However, in the case of fixed $r$ the complexity could be prohibitive due to large $k$. In that sense, the case $k > \sigma$ is more interesting for the class of asymptotically MDS array BP-XOR codes. 

\subsubsection{ Case $k > \sigma$}
With classical array BP-XOR codes, the block length $n$ is constrained by the following upper bound for $b \gg 1$,
\begin{eqnarray}
n \leq k + \sigma -1 + \left\lfloor \frac{\sigma(\sigma-1)}{k-\sigma} \right\rfloor - \textbf{1}_{(k-\sigma)|(\sigma-1)\sigma}
\end{eqnarray}
which is the same for asymptotically MDS array BP-XOR codes as mentioned in Section II. However, as the block length gets large as well, we shall no longer have constraints on the size of the block length for asymptotically MDS BP-XOR codes. 

Next, we provide another set of parameters for Mojette code that shall satisfy $k > \sigma$. The possibilities of the pair $(p_i, q_i)$ selection for making $k > \sigma$ is not unique. We will consider the typical class as given in Construction \ref{cons35}.

\begin{construction} \label{cons35}
Let us consider the following choice of coprime integers for $n$ projections, 
\begin{align}
\log_2(q_i) &= q_e, \nonumber \\
p_i &\in \mathfrak{U} = \left\{\left\lceil-n+1\right\rceil_{odd},\dots,-1,1,3,\dots,\left\lceil n-1 \right\rceil_{odd}\right\} \label{option2}
\end{align}
where $q_e$ is a positive even number, and $\left\lceil . \right\rceil_{odd}$ rounds to the next biggest odd number of the argument, respectively. 
\end{construction}

Note that using construction \ref{cons35}, it is easy to verify $GCD(p_i, q_i) = GCD(p_i, 2^{q_e}) = 1$ for non-negative $q_e$. Also, we have $k > \sigma = \max_i \{ \min\{\lceil b/|p_i| \rceil, \lceil k/|q_i| \rceil \} \} = \lceil k2^{-q_e} \rceil$. It is of interest to quantify the coding overhead to be able to find the upper bounds on the code block length. 

\begin{theorem} \label{thm39}
For the Mojette code with parameters as given in construction 3.5, for $b \gg 1$, we have
\begin{align}
\epsilon(n,b)  &\leq \frac{\lceil k2^{-q_e} \rceil}{kb}
\Big((k-1) \left(n - \frac{\lceil k2^{-q_e} \rceil}{2} \right) \dots  \nonumber \\ 
& \ \ \ \ \ \ + (b-1)2^{-q_e} + 1 \Big) - 1 + o(1) \nonumber
\end{align}
where $q_e$ is a positive integer, and $\left\lceil . \right\rceil_{odd}$ rounds to the next biggest odd integer of the argument, respectively.
\end{theorem}

\begin{proof}
See Appendix \ref{AppendixB} for the proof of this theorem.
\end{proof}

Note that as long as $q_e | \log_2(k)$, we have $\epsilon(n,b) \rightarrow 0$ for large $b$ demonstrating the asymptotically optimal overhead property. Similarly, for fixed $r$ and $b$, if $n \rightarrow \infty$ then we have $\epsilon(n,b) \rightarrow \infty$. Finally, using equation (\ref{eqn9}) we can express the upper bound on $n$ as follows,
\begin{equation} 
n \leq k +  \frac{\sigma \lceil k2^{-q_e} \rceil}{kb}
\Big((k-1) \left(n - \frac{\lceil k2^{-q_e} \rceil}{2} \right)  + (b-1)2^{q_e} + 1 \Big)  - 1
\end{equation}

We provide some numerical results that compute the upper bounds for comparison. Let us set $q_e = 1 > 0$ and assume a large $b$ value, such as $b=10000$ (this choice is completely arbitrary) and compare the upper bounds on $n$ with using classical MDS array BP-XOR codes and their asymptotically optimal version proposed in our study, abbreviated as AMDS. We present results in Fig. \ref{fig3} in which a)  demonstrates that classical MDS array BP-XOR codes are only possible for very small values of $k$. On the other hand, although the same is true for asymptotically MDS BP-XOR codes for small $k$, it is also observed that for large enough $k$ our bounds are larger than the required $n$ (fixed by the code rate), allowing possible constructions to achieve the corresponding code rate asymptotically. In Fig. \ref{fig3} b), we present the possible minimum code rate (due to the upper bound on the block length) as a function of assumed nominal code rate for different coding schemes. The region above all curves represent achievable code rates. As can be seen, with increasing $k$, AMDS provides more freedom in choosing the right code rate. Fig. \ref{fig3} a) also presents the upper bound behavior for small $k$ on the top-left corner. The plot includes a curve ``required $n$" to denote the required value for $n$ for the corresponding code rate $r=k/n$. 

In order to see clearly the range of rates that are possible with both constructions, the same figure b) depicts the minimum rate that is possible as a function of the assumed rate. Note that with asymptotically MDS array BP-XOR codes, the upper bound on $n$ depends on the coding overhead which is a function of the code rate. Thus, the minimum code rate changes as the assumed code rate changes. For each assumed rate, we calculate the upper bound and then compute the minimum code rate possible. With respect to classical MDS BP-XOR codes, since the upper bound does not change with varying assumed code rates (since the coding overhead is always zero), the curves are flat. As can be observed, the region that lies above the curves is the possible code rates. However, there is no guarantee each and every assumed rate would be achievable. However, as can be seen for large $k$, it becomes impossible to construct classical array MDS BP-XOR codes with a rate smaller than unity. In contrast, by relaxing the exact MDS condition, we can improve the region of achievability.

\section{End-to-End Latency Performance Analysis}
\label{LatencyPerf}

In this section, we shall primarily focus on the total computational latency of the proposed coded system. Then, we will shortly touch upon the communication cost and compare it with other well--known coded computation schemes in the literature. 

\subsection{Computation Latency} 
Similar to the past studies, our time analysis also focuses on exponential task time for each processor. More specifically, we choose the most basic operation to be the ``long dot product" operation ($s$ multiplications and $s-1$ additions for a typically large $s$) in our system, distributed exponentially with parameter $\mu$ i.e., having the probability density function $f(t) = \mu e^{- \mu t}$ and the cumulative distribution function (cdf) $F(t) = 1 - e^{-\mu t}$. If $i^{th}$ processor of a cluster node performs $\sigma$ such dot products, then its cdf will be $F(t/\sigma)$, a scaled version of the original distribution \cite{Lee2016}. Also, the processing power of each $\rho$ master node processor is $c$ times greater than that of the compute cluster which makes the master processor computation rate to be $c\mu$. The parameter $c$ is referred to as \textit{compute factor} for the rest of our discussion.  To minimize the workload of processors and maximize the parallelization, we assume $m \geq b$ and $b < k$ for the rest of our discussion.

For a given group of $b$ processors, the $l^{th}$ order statistics of $(T_1,\dots,T_b)$ is represented by $T_{l:b}$. The expected value of the maximum of $T_i$s each distributed exponentially with rate $\mu$ can be shown to be $\mathbb{E}[T_{b:b}] = H_b/\mu$ where $H_b = \sum_{j=1}^b 1/j$ is the $b^{th}$ harmonic number. Similarly, the expected value of the $l^{th}$ order statistics of $b$ exponential random variables of rate $\mu$ is $\mathbb{E}[T_{l:b}] = (H_b-H_{b-l})/\mu$. Note that for sufficiently large $b$, we have the approximation $H_b \approx \log(b)$.

In the uncoded case, the average latency characterization is straightforward. Since there is no encode/decode operation at the master, all it takes to compute the product is to distribute $kb$ dot products over $kb$ processors and collect the result for a successful merge. In that case, the slowest processor output will determine the expected latency for the overall product computation i.e., $\mathbb{E}[T_{uncoded}] \approx (1/\mu)\log (kb)$ for large $b$. The following theorem characterizes the asymptotical computation time of both encoding/decoding and parallel task completion for single dimensional $[nb,kb]$ MDS polynomial codes scattered across the compute cluster introduced in Fig. \ref{fig:cc1}. Throughout this section, we assumed $s=m$ without loss of generality.

\begin{theorem} \label{Thm31}
Let us use $[nb,kb]$
Polynomial code \cite{Yu2017} to distribute computation over $nb$ processors, the asymptotic latency ($b \rightarrow \infty$, $b < k$) is given by
\begin{align}
\mathbb{E}[T_{poly}]   \gtrapprox  \frac{2nb^2}{c\rho\mu}\log (\rho) + \frac{kb\log^2(kb)}{c\rho\mu} \log (\rho) + \frac{1}{\mu}  \log\left(\frac{n}{n-k}\right) 
\end{align}
where $c$ is the compute factor of the system.
\end{theorem}

\begin{proof}
Note that in general $a_j$s and $b_j$s (in our previous example) contain $(b-1)$ dot products each. The encoder performs these dot products for each processor except one (total of $nb-1$), giving us a total of $2(b-1)(nb-1) \approx 2nb^2$ dot products executed in a sequential manner. Decoding is based on the interpolation of a polynomial of degree $kb-1$ and the best known algorithm to solve this is on the order of $kb\log^2(kb)$ operations \cite{Kung}. Although operations are likely to be more than a dot product and not fully parallelizable, we estimate the complexity in this particular way to target the most favorable scenario, giving the competitors the best chance of winning. Note also that the master processing is exponentially distributed with parameter $c\mu$ due to independence and there are $\rho$ processors in operation bringing up the $\log(\rho)$ factor in the expression. On the other hand, the parallel executions perform only single dot product and any $k$ column collections i.e., $kb$ executions will suffice to recover the multiplication result, amounting to an expect delay of $\approx \frac{1}{\mu}\log(\frac{n}{n-k})$. Adding the expected encoding/decoding time and the parallel task time, the result follows.
\end{proof}

Although polynomial codes provide order optimal parallel task time, the encode/decode time shall be the bottleneck for the overall performance if $c$ (and it practically) does not scale with the increasing matrix sizes.  Later studies have shown that \textit{MatDot} codes can further improve the recovery threshold at the expense of worse $T_{sm}$ performance \cite{Fahim}. The following theorem characterizes its overall computation performance.

\begin{theorem}Let us use $[nb,kb]$ MatDot code to distribute computation over $nb$ processors, the asymptotic latency ($b \rightarrow \infty$, $b < k$) is given by
\begin{align}
\mathbb{E}[T_{MatDot}]  \gtrapprox& \frac{2nb^2}{c\rho\mu} \log (\rho) + \frac{kb(k+b)\log^2(k+b)}{c\rho\mu} \log (\rho)  \dots \nonumber \\ 
& \ \ \ \ +  \frac{1}{\mu}  \log\left(\frac{n}{n-k/b-1}\right) 
\end{align}
where $c$ is the compute factor of the system.
\end{theorem}

\begin{proof}
We realize that the encoding of MatDot codes is very similar to polynomial codes resulting in the same order number of dot products $2nb^2$ executed in a sequential manner at the master. Similarly, decoding is based on polynomial interpolation but unlike polynomial codes, it suffices to collect $k+b-1$ successful processor outputs to reconstruct the multiplication result which reduces the decoder complexity. Hence for $kb$ elements, we have on the order of $kb(k+b)\log^2(k+b)$  operations, again which may not be simple dot product. All encode/decode operations are performed by $\rho$ processors in parallel and hence the $\log(\rho)$ factor. On the other hand, the parallel executions perform only a single outer product which is at least as complex as a single dot product. In fact, it may be assumed that the outer product is at most equal to $b$ dot products. Although the exact complexity figures for the outer product can be incorporated into this expression, a simple lower bound would suffice to demonstrate the degraded total latency performance of MatDot codes compared to its competitors. Finally, since any $k+b-1$ executions will suffice to recover the the original result (a polynomial of degree $k+b-2$), it amounts to an expected parallel processing delay of $\gtrapprox  \frac{1}{\mu}\log(\frac{nb}{nb-k-b+1})$. Adding the expected encoding/decoding and the parallel task time in the same order, the result will follow.
\end{proof}

Using the same line of thought and assuming a single outer product is at most as complex as $b$ dot products, then we can asymptotically upper bound the expected latency as follows
\begin{align}
    \mathbb{E}[T_{MatDot}]  \lessapprox & \frac{2nb^2}{c\rho\mu} \log (\rho) + \frac{kb(k+b)\log^2(k+b)}{c\rho\mu} \log (\rho) \dots \nonumber \\ 
    & \ \ \ \ +  \frac{b}{\mu}  \log\left(\frac{n}{n-k/b-1}\right) 
\end{align}

We realize that even though the parallel task execution time performance of MatDot codes could be better compared to polynomial codes in the most favorable case, its total computation time is worse (in all circumstances) with scaling $k$. In addition, as we shall see later, it has worse communication cost $T_{ms}$ as well to help reduce $T_{sm}$. In both computation schemes however, the encode/decode times seem to be the main latency bottleneck, especially for small $c$ and $\rho$. 

Next, we provide the latency performance of MDS array  BP-XOR codes (AMDS) in the sublinear regime in the size of the product which distributes the encoding operation over the cluster nodes in order to achieve better end-to-end latency performance at the expense of increased bandwidth.

\begin{theorem} \label{Thm43}
Let us use $[n, k, t, b]$ MDS array BP-XOR code over $n$ nodes each with $b$ processors, for a fixed maximum node degree $\sigma$, the asymptotic latency ($k \rightarrow \infty$) is given by
\begin{align}
\mathbb{E}[T_{AMDS}] \approx \frac{\sigma k b}{c\rho\mu} \log(\rho) + \frac{\sigma}{\mu} \log\left(b \right) + \frac{b^{\sigma-1}}{\mu} \sqrt{2\log \left( \frac{n}{t} \right) }
\end{align}
\end{theorem}

\begin{proof}
See Appendix \ref{AppendixC} for the proof of this theorem. 
\end{proof}

In this expression, we primarily note that implication of $\sigma$ being constant is that the number of backup nodes is sublinear in $k$, i.e., $o(k)$. However, the number of workers, i.e., processors can still grow by increasing the parameter $b$. 
We finally note that the bound on $n-k$ given in \eqref{inequality} satisfies the constraint of intermediate order statistics \cite{Lee2017} and the average latency is linear in $k$ achieving the order optimal computation. Needless to point out that in all of these theorems, we assumed moderate $\rho$ and $c>1$ which aligns with the assumption that master nodes are typically more capable. 

Next,  we  provide the  latency  performance  of a class of MDS  array BP-XOR codes  (AMDS)  in   linear  regime  in  the  size of  the  product i.e., $\Theta(k)$. In other words, the number of stragglers increases linearly with $k$ where we assume $(1+\xi)k$ total number of nodes to compute the matrix multiplication for a fixed $\xi>0$. Note that classical MDS array BP-XOR codes cannot achieve $\xi>0$ as $k \rightarrow \infty$ since $k/n \rightarrow 1$. Hence, the following theorem applied to the asymptotical version of MDS array BP-XOR codes only.  

\begin{theorem} \label{Thm44}
Let us use a $[n, k, t, b, b^{\prime}]$ asymptotically MDS array BP-XOR code used over $n = (1 + \xi)k$ nodes for some $\xi>0$. Let also each node to be equipped with $b_i$ processors for $1 \leq i \leq n$, executing $\sigma_i$ dot products at most for a fixed node degree  $\sigma = \max\{\sigma_i\}$. If we define $x_i = \sigma_i \log(b_i)$ and $x_{i_1}$ and $x_{i_n}$  to be minimum and maximum of all $x_i$s, respectively, then the asymptotic latency ($k \rightarrow \infty$) can be upper bounded by
\begin{align}
\mathbb{E}[T_{Asy}] \leq \frac{\sigma k b^\prime}{c \rho \mu}\log(\rho) + \frac{\sigma}{\mu} \log(b^\prime)  + \frac{\overline{\sigma_b}}{\mu} \sqrt{2 \log \left(\frac{1 + \xi}{\xi}\right)}
\end{align}
where 
\begin{align}
    \overline{\sigma}_b =  
    \sqrt{\frac{n(x_{i_n} - x_{i_1})^2}{4} + \frac{1}{n} \sum_{i=1}^n b_i^{2(\sigma_i-1)} },
\end{align}
$b^\prime = b (1 + \epsilon(b,n))$ and $\epsilon(b,n)$ is the maximum coding overhead.
\end{theorem}

\begin{proof}
See Appendix \ref{AppendixD} for the proof of this theorem. 
\end{proof}

We note that with $b_i=b$ and $\sigma_i=\sigma$, this general result will be identical to the result of Theorem \ref{Thm43}. We also notice that for a fixed $b$, asymptotical version has parallel task time of $O(\sqrt{k})$ whereas original version has $O(\sqrt{\log(k^2)})$. This means that although the asymptotical version allows us better flexibility in choosing the number of stragglers in the network, due to unbalanced computation allocation among network nodes, its parallel execution becomes worse. However, the overall execution time is still linear in $k$, achieving the order optimal computation time from an end-to-end perspective. 

\subsection{Communication Costs}

In polynomial codes, after encoding operation takes place in the master node, the generator communicates $s$ symbols for both matrices ($a_j$ and $b_j$) to the processor $j$ to compute $a_jb_j$ for $j\in\{1,2,\dots,nb\}$. Since a total of $nb$ processors are used, it communicates a total of $2snb$ symbols (compare this to the uncoded case where a total of $2skb < 2snb$ symbols are communicated instead). The sink node collects only $kb$ symbols--the number dictated by the recovery threshold (the outcome of dot products) to initiate successful polynomial interpolation. 

In the case of MatDot codes, the generator communicates the same amount of information (in vector form) with the processors i.e., a total of $2snb$ symbols. However, the processors compute and communicate matrices instead of dot products. The decoder only needs to receive $k+b-1$ processor outputs for successful reconstruction, each being a matrix of size $k \times b$ i.e., a total of $kb$ symbols are communicated to the decoder. Hence the number of symbols communicated with the sink for successful decoding is $kb(k+b-1)$ which effectively boosts the $T_{sm}$ term in Eqn. \eqref{eqnTtotal}.

Using MDS array BP-XOR codes, the generator will have to communicate $2\sigma s$  symbols for each processor  in the worst case. Using a total of $nb$ processors, the generator communicates $2\sigma s n  b$ symbols, which boosts the $T_{ms}$ term in Eqn. \eqref{eqnTtotal}. Since for an order  optimal latency performance we typically choose a fixed $s$ and $\sigma \leq k + b - 1$, MDS array BP-XOR code provides better latency characteristics due to utilizing less bandwidth compared to MatDot codes on master-slave (map) link instead of slave-master (reduce) link. Finally, the sink collects at least $kb$ symbols (just like polynomial codes) to initiate the linear-time iterative decoding. 

In the case of asymptotically MDS array BP-XOR codes \cite{Suayb2018}, the generator will  communicate $2\sigma s$  symbols for each processor in the worst case. Using a total of $nb^\prime$ processors, the generator communicates a total of $2\sigma s n  b^\prime$ symbols, which similarly boosts the $T_{ms}$ term in Eqn. \eqref{eqnTtotal}. Finally, the sink collects at least $kb^\prime$ symbols (more than that of MDS array BP-XOR and polynomial codes) to initiate the linear-time iterative decoding. Note that the asymptotically MDS array codes provide more flexibility in terms of coding rate (the number of stragglers), it also leads to $2\sigma s n \epsilon(b,n)$ and $k \epsilon(b,n) b$ more symbols (compared to the non-asymptotic version) to communicate for master-slave and slave-master links, respectively.

\begin{table}[t!] 
\caption{Communication Cost (symbols) for Various Coded Computation schemes}
\label{table3a}
\centering
\begin{tabular}{c|c|c} 
    \hline
    Scheme  &  master-slave (map) & slave-master (reduce) \\
    \hline \hline
    Uncoded   &  $2skb$ & $kb$ \\ \hline
    Polynomial codes   & $2snb$ & $kb$ \\ \hline
    MatDot codes &  $2snb$ & $kb(k+b-1)$\\ \hline
    Exact MDS BP-XOR  &  $2\sigma snb$ & $kb$ \\ \hline
    Asym. MDS BP-XOR  & $2\sigma s n  b^\prime$  & $kb^\prime$ \\ \hline
\end{tabular}
\end{table}

%\hfill
\begin{table}[t!] 
\caption{Simulation Parameters}
\label{table3b}
\centering
\begin{tabular}{c|c}
    \hline
    Parameter  &  Value\\
    \hline \hline
    $b$   &   20\\ \hline
    $c$   &   50\\ \hline
    $\sigma$   &   7\\ \hline
    $\rho$   &   50\\ \hline
    $\mu$   &   1\\ \hline
\end{tabular}
\vspace{-4mm}
\end{table}

\section{Numerical Results}
\label{numerical}

Let us provide the expected end-to-end computational latency of the aforementioned coded computation schemes for matrix--matrix multiplication. We would like to remark that we only consider codes that have the MDS property in this study. This way we ensure that the computation takes place on the same number of compute nodes/processors for a given workload despite low-complexity approaches such as LDPC codes \cite{Maity2019} could be employed to reduce encoding complexity at the expense of distributing computation over many more processors without taking into account the underlying hierarchical computation architecture.  

We employ a Monte Carlo simulation to assess the average end-to-end computation time. Since we assume scaling clusters in our study, we assume $k$ to tend to large values. Unless stated otherwise, the simulation parameters summarized in Table \ref{table3b} are used and we mostly are interested in the range $b < k$. We have also set the number of stragglers to maximum possible i.e.,
\begin{eqnarray}
n - k  = \sigma -1 + \left\lfloor \frac{\sigma(\sigma-1)(b-1)}{(k-\sigma)b+ \sigma -1}\right\rfloor
\end{eqnarray}
to make sure that an appropriate selection (high rate) of an MDS array BP-XOR code can be made. As have been the measure of comparison of the past studies (since \cite{Lee2016}), we plot the expected total computation latency $\mathbb{E}[T]$ as a function of growing $k$ as illustrated in Fig. \ref{fig:fig12} a). In the uncoded case, the master node does not perform any computations and the matrix multiplication is distributed over $kb$ processors. Due to stragglers, its performance is worse than polynomial codes and MDS Array BP-XOR codes (AMDS). Another natural observation is that since encoding/decoding requirements are escalating as $k \rightarrow \infty$, the total computation latency performances of all schemes get worse. However, AMDS demonstrate an order of magnitude better latency performance compared to polynomial codes thanks to its suitable structure that allows low complexity decoding and distributed encoding. This family of codes achieves this performance at the expense of using more bandwidth between the master and slave nodes. Although MatDot codes ensure a better recovery threshold, the decoding complexity makes its end-to-end latency performance worse than its competitors.

\begin{figure*}[t!]
    \centering
    \includegraphics[width=0.5\columnwidth, height=6.5cm]{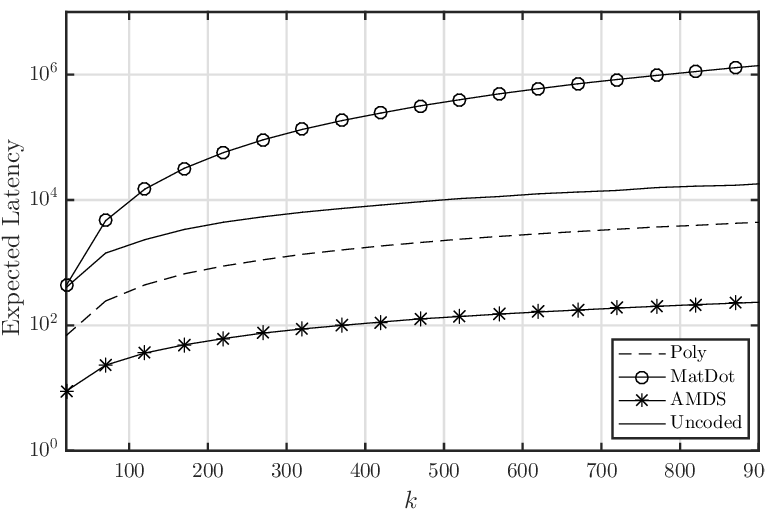}
    \includegraphics[width=0.5\columnwidth, height=6.5cm]{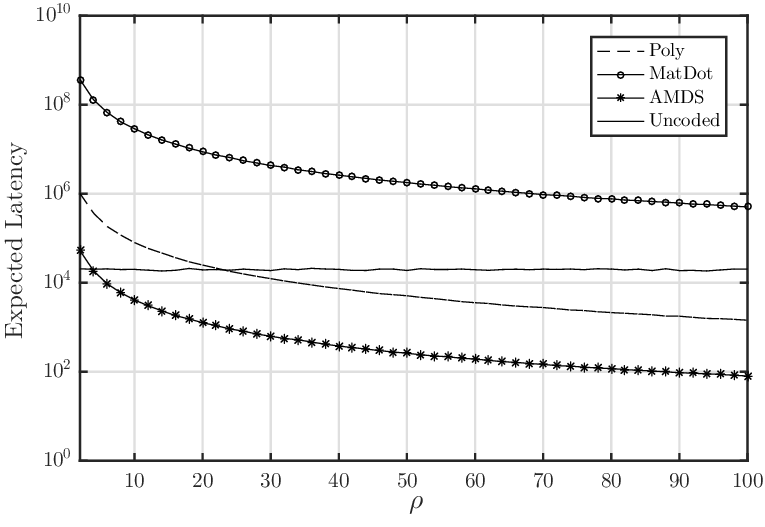}
    \caption{a) Expected total computation time of various computation schemes as a function of $k$. b) Expected total computation time of various computation schemes as a function of $\rho=c$ with fixed $k=1000$.}
    \label{fig:fig12}
    \vspace{-4mm}
\end{figure*}

\begin{figure*}
    \centering
    \includegraphics[width=0.5\columnwidth, height=6.5cm]{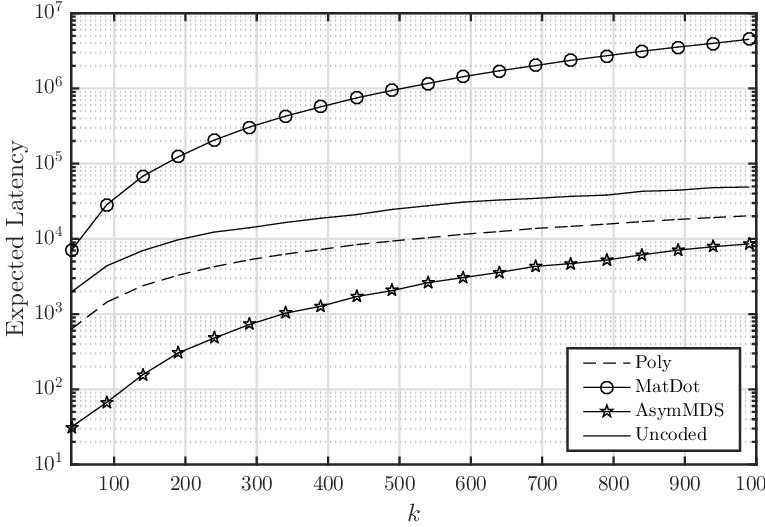}
    \includegraphics[width=0.5\columnwidth, height=6.5cm]{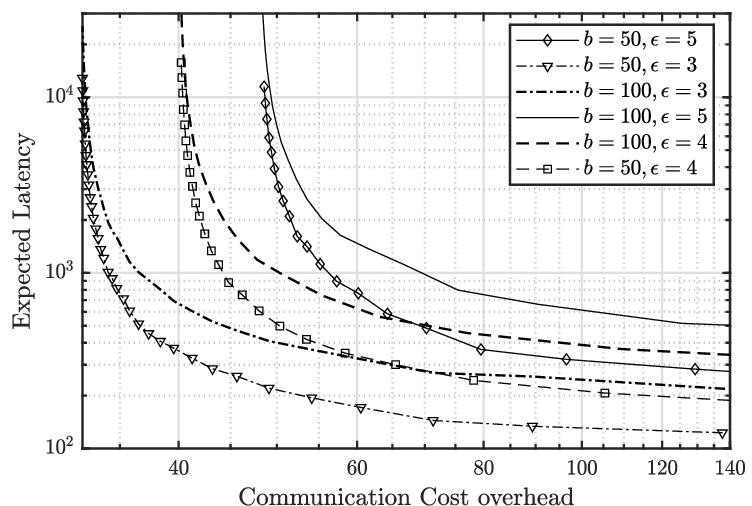}
    \caption{a) Expected total computation time of various computation schemes as a function of $k$ for 10\% straggler ratio. b) Expected total computation time of Asymtotically MDS array BP-XOR codes as a function of communication cost.}
    \label{fig:fig45}
    \vspace{-4mm}
\end{figure*}

In Fig. \ref{fig:fig12} b), we have fixed $k=1000$ and used all values given in Table \ref{table3b} except $\rho$ and $c$. We have set $\rho=c$ and varied both between 2 and 100. Hence, as we go from left to right along the abscissa, the master's parallel computation capability would increase. Note that the uncoded scheme is not affected by the master's computation capabilities as it does not call for any encoding and decoding processes. Also, fixing $k$ implies that the number of stragglers is $\sigma-1=6$ i.e., fixing $n = 1006$ meaning that the code rate $\approx 1$. For a typical choice of $\rho=c=10$, this implies that using only a few extra computations, AMDS coding scheme provides almost five times (20000 versus 4108) better total computation latency compared to the uncoded scheme. 

One of the limitations of the original AMDS code is that the number of straggler tolerance is sublinear in $k$ as the cluster scales. In other words, the previous numerical setting such as $n=1006$ implied that the system is only tolerant to six stragglers since $k=1000$. This might be a quite limitation of the code's usability for real clusters in which the number of stragglers may typically scale with the number of nodes \cite{Schroeder2010}, \cite{Barroso}. 

Let us suppose we have a scaling cluster with a fixed $10\%$ straggler ratio which is shown to be quite typical in high-performance computing systems \cite{Schroeder2010}. Note that in order to generate redundant computation for scaling stragglers using AMDS codes, we need to tolerate extra computation overhead. We plot In Fig. \ref{fig:fig45} a) the expected latency as a function of $k$ for various coding schemes. We have also included uncoded performance as the baseline for comparison.  As can be seen, the asymptotical version (AsymAMDS) provides an order of magnitude better expected latency compared to Polynomial codes for small size clusters ($k \approx 200$). However, as the size of the cluster increases, mainly due to increased coding overhead $\epsilon(b,n)$, the computation overhead becomes a sublinear function of $k$, making the overall encode/decode process a non-linear function of $k$. This is why for large $k$ (for instance $k \approx 1000$), the latency performance gets closer to that of Polynomial codes.

\begin{table*}[t!]
\caption{Comparison and Break-down of the Relative Overall Operation Time of Different Coding Schemes for the Parameter Tuple ($b=100$, $\epsilon=3$, $\rho=c=50$) and An Average Time Unit $\mu=1$.}
\centering
\label{table444}
\begin{tabular}{|c|c|c|c|c|c|c|c|}
\hline
Method                    & $k$ & $t$ & Encode@Master & Decode@Master & Cluster Time & $\mathbb{E}[T]$ & Comm. Cost overhead \\ \hline
\multirow{3}{*}{Polynomial}     & 100  &  37 & 4924.9 & 1524.6 & 1.3 &  6450.8 &  0.37        \\ \cline{2-8} 
                          & 1000  & 27 & 37003.9 & 23844.5 & 3.64 &  60852.04       &  0.027        \\
                          \cline{2-8} 
                          & 10000  & 27 & 361689.5 & 343903.7 & 5.9 & 705599.1       &  0.0027        \\
                          \hline
\multirow{3}{*}{MatDot}   & 100  &  37 & 4925.5 & 100809.7 & 1.46 &  105736.66 &   198.4       \\ \cline{2-8} 
                          & 1000  & 27 & 36917.2 & 9746604 & 1.1 &   9783522.3      &   1098  \\
                          \cline{2-8} 
                          & 10000  & 27 & 361558.4 & 1543655907 & 1 &   1544017466.4      &   10098 
                          \\ \hline
\multirow{3}{*}{AMDS}     & 100  &  6  & 0 & 126 & 51.5 &     177.5      &   10.6       \\ \cline{2-8} 
                          & 1000  & 6  & 0 & 1262.2 & 67.5 &    1329.7    &    8.19   \\ \cline{2-8} 
                          & 10000  & 6  & 0 & 12520 & 83.6 &   12603.6   &    8.02    \\ \hline
\multirow{3}{*}{AsymAMDS} & 100  &  37  &0 & 496.3 & 49.6 &   545.9       &   42.4       \\ \cline{2-8} 
                          & 1000  & 27 & 0 & 4949 & 67.1 &   5016.1       &      32.7  \\ \cline{2-8} 
                          & 10000  & 27 & 0 & 50338.6 & 83.2 &   50421.8       &      32.07   \\ \hline
\end{tabular}
\end{table*}

Generally speaking, the amount of time it takes for a computation to complete isn't the only important factor. We must take communication latency into account, as well as the bandwidth traffic generated during computing. To do so, we present the resulting tradeoff between the end-to-end computation latency and communication cost overhead. To be able to numerically present the communication cost overhead, we divide the total communication cost by the cost of an uncoded case and subtract one. For instance, using AsymAMDS codes, we divide the master-slave communication cost of $2\sigma snb^\prime$ by $2skb$ and subtract one, which results in $\sigma n (1 + \epsilon(n,b))/k - 1$. Following a similar logic, the total communication cost overhead (including both master--slave and slave--master communications) can be given as $\sigma n/k  + (\sigma n/k+1)\epsilon(n,b)$. Since the communication cost of Polynomial codes is minimal and that of MatDot codes scales with $O(k^2)$, we only illustrate the tradeoff curve using AsymAMDS codes for brevity. We use $\epsilon(b,n) \in \{3,4,5\}$ and $b \in \{50,100\}$ by varying $k$ values (and indirectly $t$ due to the bound in \eqref{eqn3}) over the range between 40 and 5000. As can be seen from Fig. \ref{fig:fig45} b), the relationship is inverse and as we allow more communication between the master and slave nodes, we obtain better average access latency performance. One of the observations is that for a fixed $b$, as the coding overhead increases, the trade-off curve shifts right which increases the communication cost overhead. On the other hand, for a fixed coding overhead, as $b$ increases we also observe that expected total computation time increases as well. Hence, we desire minimal $b$ values. However, choosing these values as small as possible would significantly limit the achievable code rate and eventually reduce the percentage of stragglers that can be tolerated in the cluster. 

In Table \ref{table444}, we have finally provided the breakdown of the latency figures between the encode time spent in the master, decoding time and finally the time spent in the cluster nodes. We have assumed $b=100, \epsilon = 3, \sigma =7, c=\rho=50, \mu=1$ and an outer product to be approximately equal to $b$ dot products for MatDot codes. We have conducted over 10000 simulations and reported the averages. As can be seen, a similar trade-off between $\mathbb{E}[T]$ and the communication cost overhead can be observed. Also, we can realize that cluster time with the proposed coding schemes increase compared to Polynomial and MatDot codes due to distributed encoding and multiple ($\sigma$) dot product computations per node. Although AsymAMDS performs worse than AMDS in terms of latency and communication cost, it allows more flexibility for the selection of the right code rate making it suitable for scaling stragglers scenarios.

\section{Conclusion}
\label{conc11}

A fault-tolerant massive matrix product scheme is presented under a realistic hierarchical compute cluster model using MDS array BP-XOR codes. The implications of the limitations on the maximum block length of such codes constrain to a fixed size of stragglers. On the other hand for scaling stragglers, an asymptotic version based on projection geometry is proposed to provide an efficient solution to the massive matrix multiplication process in which stragglers also scale. The proposed scheme has a few novelties: (1) it allows the computation of encoding to be distributed over the cluster nodes at the expense of increased communication cost, (2) it has an extremely efficient decoding process based on pure XOR logic. Furthermore, (3) it can be used as component codes of $d$-dimensional product coding schemes to allow for more powerful coded computations. Finally, (4) due to the iterative nature of decoding, this coding scheme is one of the best candidates for future master-less computation frameworks. Parallelization of the iterative decoding process over the slave nodes is quite possible and will be investigated in later work. One of the other ongoing works is the minimization of the communication cost through intelligent compression, due to offloading encoding operation to cluster nodes.

\appendices
\section{Proof of Theorem \ref{Thm37}} \label{AppendixA}
Let us start by defining the following utility function,
\begin{eqnarray}
\varphi(x) = \left\lfloor\frac{x}{2}\right\rfloor\left(\left\lfloor\frac{x}{2}\right\rfloor + 1\right) \textmd{ for $x$ } \geq 0.
\end{eqnarray}
Also let $I_t = \{0,1,\dots,t-1\}$. Using these definitions, we state the following lemma next.

\emph{Lemma A.1:} For the projection set given as in (\ref{option1}), we have the sum $\sum_{i=0}^{t-1} |p_i|$ that can be expressed in a closed form using the utility function
\[
\sum_{i \in I_t} |p_i| = \frac{1}{2} \left( \varphi(t) + \varphi(t-1) \right) =
\begin{cases}
\frac{t^2-1}{4}, & \textrm{if $t$ is odd} \\
\frac{t^2}{4}, & \textrm{if $t$ is even}
\end{cases}
\]
This lemma can easily be proved by considering $t$ odd and even cases using induction, separately. Note that the integer sequence $\sum_{i \in I_t} |p_i|$ is given by \emph{A002620} \cite{OEIS}. Using this result, for a given pair of projections $t_2$ and $t_1$ satisfying $t_2 > t_1$, with the associated projection parameters $(p_i^{(t_2)}, q_i^{(t_2)} = 1)$ and $(p_i^{(t_1)}, q_i^{(t_1)} = 1)$ selected based on construction \ref{Cons33} (Eqn. (\ref{option1})), we can deduce that
\begin{eqnarray}
\frac{t_2^2 - t_1^2 - 1}{4} \leq \sum_{i=0}^{t_2-1} |p_i^{(t_2)}| - \sum_{j=0}^{t_1-1} |p_j^{(t_1)}| \leq \frac{t_2^2 - t_1^2 + 1}{4} \label{eqnsum}
\end{eqnarray}

Note that since $q_i=1$, it is sufficient to collect $k$ projections for perfect reconstruction. Thus, the upper/lower bounds given in equation (\ref{eqnsum}) are particularly useful if we set $t_2 = n$ and $t_1 = n - k$ to be able find the contributions from the largest $k$  projections in the sum that appears in the worst case coding overhead expression. Let $i^\prime$ be the index such that $p_{i^\prime}^{(t_2)} = p_0^{(t_1)}$ and define the set
\begin{eqnarray}
S = \{i^\prime, i^\prime + 1, \dots, i^\prime + n - k - 1\}
\end{eqnarray}

The worst case coding overhead in this case is given by the following
\begin{align}
\epsilon(n,b) &= \frac{1}{kb} \left( \sum_{i \in I_t \backslash S} |p_i| (k-1) + |q_i| (b-1) + k \right) - 1 \\
&= \frac{(b-1)k + \frac{k-1}{2}\left(\varphi(n) + \varphi(n-1)\right)}{kb} \nonumber  \\
& \ + \frac{-\frac{k-1}{2}\left( \varphi(n-k) + \varphi(n-k-1)\right) + k}{kb}  - 1 \label{eqn100} \\
&= \frac{k-1}{2kb}\left(\varphi(n) + \varphi(n-1) - \varphi(n-k) - \varphi(n-k-1)\right)
\end{align}
where Eqn. (\ref{eqn100}) follows from the Lemma \emph{A.1}.
Again, using Lemma \emph{A.1} and Equation (\ref{eqnsum}), and through some algebra, we can bound the worst case coding overhead as follows,
\begin{eqnarray}
\frac{k-1}{4kb}\left(2kn-k^2 - 1\right)
\leq \epsilon(n,b) \leq \frac{k-1}{4kb}\left(2kn-k^2 + 1\right)
\end{eqnarray}
which can  be accurately approximated for $b \gg 1$ as
\begin{eqnarray}
\epsilon(n,b) \approx \frac{k-1}{4kb}\left(2kn-k^2\right)  = \frac{k-1}{4b}(2n-k)\label{eqn13}
\end{eqnarray}
from which the result follows.

\section{Proof of Theorem \ref{thm39}} \label{AppendixB}

Let us start with the following lemma which shall be useful to prove the theorem. 

\emph{Lemma B.1:} For the projection set given as in (\ref{option2}) with $t$ projections, we have the sum $\sum_{i=0}^{t-1} |p_i|$ that can be expressed in a closed form using the utility function
\[
\sum_{i=0}^{t-1} |p_i| =
\begin{cases}
\frac{t^2+1}{2}, & \textrm{if $t$ is odd} \\
\frac{t^2}{2}, & \textrm{if $t$ is even}
\end{cases}
\]

\begin{proof}
Let us consider the sum for even and odd $t$ separately. First we assume $t$ to be odd. Let us define the set
\begin{align}
\mathfrak{U}_{a} = \left\{\left\lceil-t+1\right\rceil_{odd}-a,\dots,-1-a,1-a,\dots,\left\lceil t-1 \right\rceil_{odd}-a\right\}
\end{align}
and notice that $\mathfrak{T} = \mathfrak{U}_0 \cup \mathfrak{U}_{1}$. Since these sets are disjoint, we have
\begin{eqnarray}
\sum_{i \in \mathfrak{T}} |p_i| = \sum_{i \in \mathfrak{U}_0} |p_i| + \sum_{i \in \mathfrak{U}_{1}} |p_i| = 2\sum_{i \in \mathfrak{U}_{1}} |p_i| + 1
\end{eqnarray}
Using this relationship and the result of Lemma \textit{A.1}, we can express
\begin{align}
\sum_{i \in \mathfrak{U}_0} |p_i| &= \sum_{i \in \mathfrak{U}_1} |p_i| + 1  
= \frac{\sum_{i \in \mathfrak{T}} |p_i| - 1}{2} + 1 \nonumber \\
&= \frac{\frac{1}{2}(\phi(2t)-\phi(2t-1))-1}{2} + 1 = \frac{t^2+1}{2}.
\end{align}

Now let us assume $t$ to be even. For this particular assumption we can rewrite $\mathfrak{T} = \mathfrak{U}_0 \cup \mathfrak{U}_{1} \cup \{t\}$. Using this observation and the result of Lemma \textit{A.1}, we can express
\begin{align}
\sum_{i \in \mathfrak{U}_0} |p_i| &= \sum_{i \in \mathfrak{U}_1} |p_i| = \frac{\sum_{i \in \mathfrak{T}} |p_i| - t}{2} \nonumber \\
&= \frac{1}{2} \left( \frac{(2t+1)^2 - 1}{4} - t \right) = \frac{t^2}{2}
\end{align}
which completes the proof of the lemma.
\end{proof}

According to Theorem \ref{katz}, we need to have $\sum_{i=0}^{t-1}|q_i| = t2^{q_e} \geq k$. This implies $t = \lceil k2^{-q_e} \rceil$ projections are sufficient for perfect reconstruction. For a given pair of projections $t_2$ and $t_1$ satisfying $t_2 > t_1$, with the associated projection parameters $(p_i^{(t_2)}, q_i^{(t_2)} = 2^{q_e})$ and $(p_i^{(t_1)}, q_i^{(t_1)} = 2^{q_e})$ selected based on construction 3.6, we can deduce that
\begin{eqnarray}
\frac{t_2^2 - t_1^2 - 1}{2} \leq \sum_{i=0}^{t_2-1} |p_i^{(t_2)}| - \sum_{j=0}^{t_1-1} |p_j^{(t_1)}| \leq \frac{t_2^2 - t_1^2 + 1}{2} \label{eqnsumX}
\end{eqnarray}

To be able find the contributions from the largest $\lceil k2^{-q_e} \rceil$  projections, we set $t_2 = n$ and $t_1 = n - \lceil k2^{-q_e} \rceil$. Using similar arguments to Appendix A, we can express the worst case coding overhead in this case as follows
\begin{align}
\epsilon(n,b) &= \frac{1}{kb} \left( \sum_{i \in I_t \backslash S} |p_i| (k-1) + |q_i| (b-1) + \lceil k2^{-q_e} \rceil \right) - 1 \\
&= \frac{(b-1)2^{q_e} \lceil k2^{-q_e} \rceil + \frac{k-1}{2}\left(\varphi(n) + \varphi(n-1)\right)}{kb}  \nonumber \\ 
& \ \ \ \ - \frac{\frac{k-1}{2}\left( \varphi(n-\lceil k2^{-q_e} \rceil) + \varphi(n-\lceil k2^{-q_e} \rceil-1)\right)}{kb}  \nonumber \\ 
& \ \ \ \ + \frac{\lceil k2^{-q_e} \rceil}{kb} - 1
\label{eqn10}
\end{align}
Using equation (\ref{eqnsumX}) and $b \gg 1$, we can accurately approximate the worst case coding overhead as follows,
\begin{align}
\epsilon(n,b)   &\leq  
\frac{\lceil k2^{-q_e} \rceil}{kb}
\Bigg((k-1) \left(n - \frac{\lceil k2^{-q_e} \rceil}{2} \right) \dots \nonumber \\
& \ \ \ \ \ \ \ \ + (b-1)2^{q_e} + 1 \Bigg) - 1 + o(1), \nonumber 
\end{align} 
which completes the proof. 

\section{Proof of Theorem \ref{Thm43}} \label{AppendixC}

In the proposed scheme, encoding does not involve any dot products at the master but it uses more bandwidth for (in terms of symbols) communication. Let us first consider the parallel task time in which we need to consider the distribution of order statistics rather than expectation. We initially consider a single network node, then we extend our analysis to expected order statistics for multiple nodes. Unlike expectation, the distribution of order statistics is more challenging. 

For simplicity, we consider asymptotics and recognize that the $j^{th}$ order statistics of $T_i$s i.e., $T_{j:z}$ (for $z$ cores or processors) converges in distribution to a Gaussian if $j/z \rightarrow 1$\footnote{This is also known as intermediate order statistics \cite{Lee2017}.} (where we define below the random variable $Y_i$ for convenience). The distribution of $Y_i$ can be expressed as
\begin{eqnarray}
Y_i := T_{j:z} \xrightarrow{d} \mathcal{N}\left( \frac{\sigma}{\mu} \log\left(\frac{z}{z-j}\right), \frac{z-j+1}{z^2 f^2(\frac{\sigma}{\mu} \log(\frac{z}{z-j}))} \right)
\end{eqnarray}
which shall model the delay for the $i$th cluster node with $b>1$ processors. From the cluster point of view and perfect reconstruction, we need any $k$ nodes completing their assigned task to be able to reconstruct $\mathbf{A}^\intercal \mathbf{B}$.

We recall $i$th expected order statistics from \cite{Royston1982} and through some algebraic manipulations, we can reach at
\begin{align}
\mathbb{E}[Y_{i:n}] \approx \frac{\sigma}{\mu} \log \left( \frac{z}{z-j} \right) + \frac{\Phi^{-1}\left(\frac{i-\alpha}{n-2\alpha+1}\right) \sqrt{z-j+1}}{z f\left(\frac{\sigma}{\mu} \log \left( \frac{z}{z-j} \right) \right)} 
\end{align}
where $\alpha = 0.375$. For $b$ processors, we need to wait until all processors complete their job. Also, we need $k$ nodes to finish their task before reconstruction. Thus, we consider $T_{b:b}$ as the limiting case which leads to 
\begin{align}
\mathbb{E}[Y_{k:n}] \approx \frac{\sigma}{\mu} \log \left( b \right) + \frac{\Phi^{-1}\left(\frac{k-\alpha}{n-2\alpha+1}\right)}{b f\left(\frac{\sigma}{\mu} \log \left( b \right) \right)}  \label{eqn7a}
\end{align}

We note that the cumulative distribution function of the standard normal $\Phi(x)$ can be written as an infinite sum and can be approximated (for $x \rightarrow \infty$)
\begin{align}
\Phi(x) &= 1 - \frac{e^{-x^2/2}}{2x\sqrt{\pi}} \left(1 - \frac{1}{x^2} + \frac{1.3}{x^4} - \frac{1.3.5}{x^6} \cdots \right) \nonumber \\
&\approx 1 - \frac{e^{-x^2/2}}{2x\sqrt{\pi}}
\end{align}
from which it follows that
\begin{align}
\log(1-\Phi(x)) \approx - \log(2x\sqrt{\pi}) - x^2/2 \label{eqn10}
\end{align}

Let us define $x := \sqrt{-2\log(1-y)}$ where $y \rightarrow 1^{-}$ which shall make $x \rightarrow \infty$.  By replacing $x$ in equation \eqref{eqn10}, this would imply that we have
\begin{eqnarray}
\Phi^{-1}(y) \approx \sqrt{-2\log ( 1-y)} = \sqrt{2 \log \left( \frac{1}{1-y} \right)} \label{eqn11}
\end{eqnarray}

Next, we note that $y = \frac{k-\alpha}{n-2\alpha +1} \rightarrow 1^{-}$ as $k \rightarrow \infty$ and can use equation \eqref{eqn11} to approximate 
\begin{eqnarray}
\Phi^{-1} \left(\frac{k-\alpha}{n-2\alpha+1}\right) \approx \sqrt{2\log \left(  \frac{n}{ n - k }\right)} \label{eqn12a}
\end{eqnarray}

Similarly, the denominator in \eqref{eqn7a} can be expressed as $b f\left(\frac{\sigma}{\mu} \log \left( b \right) \right) = \mu b^{1-\sigma}$. Finally, the decoding performs $\approx \sigma k b$ dot products sequentially in the worst case, leading to an average delay of $\frac{\sigma k b}{c\mu}$. However, if we perform these operations in $\rho$ parallel processors, and we need to wait for all operations to finish, we would obtain a delay of $\frac{\sigma k b}{cp\mu}\log(\rho)$. Adding the expected decoding time and the parallel task time with $t = n-k$ due to block MDS property, the result follows.

\section{Proof of Theorem \ref{Thm44}}
\label{AppendixD}

Similar to non--asymptotical version, let us first consider the parallel task time where we shall assume Gaussian approximation for the distribution of order statistics for analytical tractability. We initially consider a single network node, then we extend our analysis to expected order statistics for multiple nodes.

Based on intermediate order statistics i.e., $j/b_i \rightarrow 1$ and using the same notation for the random variable that characterizes the total latency for the $i$th node $Y_i$, we shall have 
\begin{align}
Y_i := T_{j:b_i} \xrightarrow{d} \mathcal{N}\left( \frac{\sigma_i}{\mu} \log\left(\frac{b_i}{b_i-j}\right), \frac{b_i-j+1}{b_i^2 f^2(\frac{\sigma_i}{\mu} \log(\frac{b_i}{b_i-j}))} \right)
\end{align}
which shall model the delay for the $i$th cluster node with $b_i>1$ processors, each executing at most $\sigma_i$ dot products and $\mathcal{N}(.,.)$ denotes the Gaussian distribution, where $\xrightarrow{d}$ means ``converges in distribution". For perfect reconstruction, we need any $k$ nodes completing their assigned task to be able to reconstruct $\mathbf{A}^\intercal \mathbf{B}$.

Through some algebra, the $k$th expected order statistics can be found to be of the form
\begin{align}
\mathbb{E}[Y_{k:n}] \approx  \mu_b + \sigma_b \Phi^{-1}\left(\frac{k-\alpha}{n-2\alpha+1}\right) 
\end{align}
where $\alpha = 0.375$, $\mu_b = \frac{1}{n\mu}\sum_{i=1}^n  \sigma_i \log \left( b_i \right)$ and
\begin{eqnarray}
\sigma_b = \sqrt{\frac{1}{n}\sum_{i=1}^n \left[ \frac{1}{b_i^2 f^2\left(\frac{\sigma_i}{\mu} \log (b_i) \right)} + \left(\frac{\sigma_i}{\mu}\log \left( b_i \right) - \mu_b \right)^2  \right]} \label{eqn18}
\end{eqnarray}
where these results easily follow using the normal analysis for order statistics in \cite{Royston1982} and for general distributions, the upper bound on the expected value of the $k$th order statistics as given in \cite{Arnold1979}. Note that setting $b_i = b$ (and hence $\sigma_i = \sigma$) for all $i$ satisfying $1 \leq i \leq n$ will lead to \eqref{eqn7a}. 

Let $x_i = \sigma_i \log(b_i)$ and $x_{i_1} \leq x_{i_2} \leq \dots \leq x_{i_n}$ be the ordered set with $x_{a} = 0.5(x_{i_1} + x_{i_n})$. Based on this definition, We realize that we can bound the square term in equation \eqref{eqn18} as follows
\begin{align}
\left(\frac{\sigma_i}{\mu}\log \left( b_i \right) - \mu_b \right)^2 \leq \frac{1}{\mu^2} \sum_{i} ( x_i - x_{a})^2 \leq \frac{n(x_{i_n} - x_{i_1})^2}{4\mu^2}
\end{align}
where the latter inequality follows because $|x_i - x_{a}|$ is at most $(x_{i_n} - x_{i_1})/2$. This leads to
\begin{align}
    \sigma_b \leq \frac{1}{\mu} 
    \sqrt{\frac{n(x_{i_n} - x_{i_1})^2}{4} + \frac{1}{n} \sum_{i=1}^n b_i^{2(\sigma_i-1)} }
\end{align}

On the other hand, using Jensen's inequality we can upper bound $\mu_b$ given by the sequence of inequalities
\begin{align}
\frac{1}{\mu} \log\left( \frac{1}{n} \sum_i b_i^{\sigma_i}\right) \nonumber  
    \geq \frac{1}{\mu n} \sum_i  \log(b_i^{\sigma_i}) = \mu_b
\end{align}
which leads to 
\begin{eqnarray}
\frac{\sigma_{\max}}{\mu} \log(b^\prime) \geq \mu_b \geq \frac{\sigma_{\min}}{n\mu} \sum_i \log(b_i)
\end{eqnarray}
where $\sigma_{\max} = \max\{\sigma_i\}$ and $\sigma_{\min} = \min\{\sigma_i\}$. We also use $\sigma = \sigma_{\max}$ to be compatible with the main text of the paper.

Similar to the approximation given in \eqref{eqn12a} for large $k$, we can express
\begin{align}
\mathbb{E}[Y_{k:n}] \approx  \mu_b + \sigma_b \sqrt{2 \log \left(\frac{1 + \xi}{\xi}\right)}
\end{align}

Finally, using the bound derived earlier for $\mu_b$ and $\sigma_b$, we can rewrite
\begin{align}
\mathbb{E}[Y_{k:n}] \leq  \frac{\sigma}{\mu} \log(b^\prime)  + \frac{\overline{\sigma_b}}{\mu} \sqrt{2 \log \left(\frac{1 + \xi}{\xi}\right)}
\end{align}
where 
\begin{align}
    \overline{\sigma}_b =  
    \sqrt{\frac{n(x_{i_n} - x_{i_1})^2}{4} + \frac{1}{n} \sum_{i=1}^n b_i^{2(\sigma_i-1)} }
\end{align}

On the other hand, the decoding performs $\approx \sigma k b^\prime$ dot products sequentially in the worst case, leading to an average delay of $\frac{\sigma k b^\prime}{c\mu}$. However, if we perform these operations in $\rho$ parallel processors, and since we would need to wait for all operations to finish, the latency would be $\frac{\sigma k b^\prime}{cp\mu}\log(\rho)$. Adding the expected decoding time and the parallel task time, the result follows.

%\section*{Acknowledgment}

\ifCLASSOPTIONcaptionsoff
  \newpage
\fi



%\bibliographystyle{IEEEtran}


%\begin{IEEEbiography}{Suayb S. Arslan}
%\end{IEEEbiography}

\begin{IEEEbiographynophoto}{Suayb S. Arslan} (Member, IEEE) received the B.Sc. degree in electrical and electronics engineering from Bo\u{g}aziçi University, Istanbul, Turkey, in 2006, an the M.Sc. and Ph.D. degrees in electrical engineering
from the University of California, San Diego, CA, USA, in 2009 and 2012, respectively. In 2009, he was with Mitsubishi Electric Research Laboratory, Boston, MA, USA, where he was involved in 
research and development of image processing and 
machine learning algorithms for biomedical applications. In 2011, he joined Quantum Corporation, 
Irvine, CA, USA, where he conducted research with IBM and HP on advanced data detection and 
coding algorithms as well as reliability modeling for increased capacity cold 
storage and cloud systems. He is currently affiliated  with MEF 
University as an associate professor, Istanbul, Turkey. His research interests include digital communication and 
storage systems, information and reliability theory, image/video processing, 
decentralized ecosystems, and the Internet of Things. He is currently serving 
as an Associate Editor for Internet of Things (Elsevier) journal.
\end{IEEEbiographynophoto}

%\newpage


%\vfill


\end{document}